\documentclass[final,1p,times]{elsarticle}
\biboptions{sort&compress}

\usepackage{amssymb}

\usepackage{amsmath}
\usepackage{amsthm}
\usepackage{amsfonts}
\usepackage{hyperref}
\usepackage{bm}
\usepackage{xspace}
\usepackage{xcolor}

\usepackage{multirow}
\usepackage{pgfplots}
\pgfplotsset{compat=newest,compat/show suggested version=false}
\pgfplotsset{
    legend image with text/.style={
        legend image code/.code={%
            \node[anchor=center] at (0.3cm,0cm) {#1};
        }
    },
}
\makeatletter
\newcommand*{\overtabline}{%
  \noalign{%
    \vskip-.5\dimexpr\ht\@arstrutbox+\dp\@arstrutbox\relax
    \vskip-.2pt\relax
    {\color{red}\hrule height 0.8pt}
    \vskip-.2pt\relax
    \vskip+.5\dimexpr\ht\@arstrutbox+\dp\@arstrutbox\relax
  }%
}
\makeatother
\newtheorem{theorem}{Theorem}[section]
\newtheorem{lemma}[theorem]{Lemma}
\newtheorem{corollary}[theorem]{Corollary}

\newtheorem{proposition}[theorem]{Proposition}
\newtheorem{definition}[theorem]{Definition}

\newtheorem{example}[theorem]{Example}
\newtheorem{remark}[theorem]{Remark}

\usepackage{mathdots,mathtools}
\usepackage[algo2e,boxed,vlined,linesnumberedhidden,
titlenumbered]{algorithm2e}
\SetKw{KwFrom}{from}
\SetKwFor{For}{For}{do}{End for}
\SetKwFor{Forall}{For all}{do}{End for}
\SetKw{KwRet}{Return}
\SetKw{KwErr}{Error}
\SetKw{KwAnd}{and}
\SetKw{KwOr}{or}
\SetKw{KwBreak}{break}
\SetKw{KwContinue}{continue}
\SetKw{KwTrue}{true}
\SetKw{KwFalse}{false}
\SetKw{KwNext}{next}
\SetKwFor{While}{While}{do}{End do}
\SetKwIF{If}{ElseIf}{Else}{If}{then}{Else if}{Else}{End if}

\SetCommentSty{mycommfont}
\usepackage{kbordermatrix}
\renewcommand{\kbldelim}{(}
\renewcommand{\kbrdelim}{)}
\usepackage{ifthen}

\newcommand\pare[1]{\left(#1\right)}
\newcommand\brac[1]{\left[#1\right]}
\newcommand\curl[1]{\left\{#1\right\}}
\newcommand\mcal{\mathcal}
\newcommand\mbf{\mathbf}
\newcommand\mbb{\mathbb}
\newcommand\msf{\mathsf}
\newcommand\sM{\msf{M}}
\newcommand\mrm{\mathrm}
\newcommand\veps{\varepsilon}
\newcommand\pos{\geq 0}

\newcommand\bt{\bm{t}}
\newcommand\bx{\bm{x}}
\newcommand\by{\bm{y}}

\newcommand\ba{\bm{a}}
\newcommand\bb{\mbf{b}}

\DeclareMathAlphabet{\mbfsf}{\encodingdefault}{\sfdefault}{bx}{n}
\newcommand\bin{\msf{b}}
\newcommand\bbin{\mbfsf{b}}
\newcommand\bi{\bm{i}}
\newcommand\bj{\bm{j}}
\newcommand\bk{\bm{k}}
\newcommand\bl{\bm{\ell}}

\newcommand\bs{\bm{s}}

\newcommand\bv{\bm{v}}

\newcommand\bw{\bm{w}}
\newcommand\bzero{\bm{0}}
\newcommand\seq{v}
\newcommand\bseq{\bv}
\newcommand\seqq{w}
\newcommand\bseqq{\bw}
\newcommand\gessel{g}
\newcommand\bgessel{\bm{g}}
\newcommand\br{\mbf{r}}
\newcommand\bone{\mbf{1}}

\newcommand\bgamma{\boldsymbol\gamma}

\newcommand\N{\mbb{N}}
\newcommand\Z{\mbb{Z}}
\newcommand\K{\mbb{K}}

\newcommand\F{\mbb{F}}
\newcommand\Q{\mbb{Q}}

\newcommand\cG{\mcal{G}}
\newcommand\cH{\mcal{H}}
\newcommand\cA{\mcal{A}}

\newcommand\cC{\mcal{C}}

\newcommand\cS{\mcal{S}}
\newcommand\cT[1][]{
  {\ifthenelse{\equal{#1}{}}{\mcal{T}}{\mcal{T}\!\!\pare{#1}}}}

\newcommand\T{\mrm{T}}

\newcommand\ceil[1]{\left\lceil#1\right\rceil}
\newcommand\floor[1]{\left\lfloor#1\right\rfloor}
\newcommand\ideal[1]{\left\langle#1\right\rangle}

\newcommand\gb{Gr\"obner basis\xspace}
\newcommand\gbs{Gr\"obner bases\xspace}
\newcommand\sgb{sparse Gr\"obner basis\xspace}
\newcommand\sgbs{sparse Gr\"obner bases\xspace}
\newcommand\Crel{C-relation\xspace}
\newcommand\Crels{C-relations\xspace}
\newcommand\Cfin{C-finite\xspace}
\newcommand\Prel{P-relation\xspace}
\newcommand\Prels{P-relations\xspace}
\newcommand\Pfin{P-finite\xspace}

\newcommand\bob{border basis\xspace}

\newcommand\Fquatre{\textsc{\texorpdfstring{F\textsubscript{4}}{F4}}\xspace}
\newcommand\Fcinq{\textsc{\texorpdfstring{F\textsubscript{5}}{F5}}\xspace}

\newcommand\FGLM{\textsc{FGLM}\xspace}
\newcommand\sFGLM{\textsc{Scalar-FGLM}\xspace}
\newcommand\lsFGLM{\textsc{Lattice Scalar-FGLM}\xspace}
\newcommand\asFGLM{\textsc{Adaptive Scalar-FGLM}\xspace}
\newcommand\lasFGLM{\textsc{Lattice Adaptive Scalar-FGLM}\xspace}
\newcommand\spFGLM{\textsc{Sparse-FGLM}\xspace}

\newcommand\agbb{\textsc{Artinian Gorenstein \bob}\xspace}

\newcommand\bms{Berlekamp--Massey--Sakata\xspace}
\newcommand\bmy{Berlekamp--Massey\xspace}
\newcommand\bln{Beckermann--Labahn\xspace}
\newcommand\ie{\mbox{{i.e.}}\xspace}

\newcommand\wrt{\mbox{wrt.}\xspace}

\newcommand\msolve{\texttt{msolve}\xspace}
\DeclareMathOperator\DRL{\textsc{drl}}
\DeclareMathOperator\LEX{\textsc{lex}}

\newcommand\adots{\mathinner{%
  \mkern1mu\raise1pt\hbox{.}%
  \mkern2mu\raise4pt\hbox{.}%
  \mkern2mu\raise7pt\vbox{\kern7pt\hbox{.}}\mkern1mu}}

\DeclareMathOperator\rank{rank}
\DeclareMathOperator\LM{\textsc{lm}_{\prec}}

\DeclareMathOperator\LC{\textsc{lc}_{\prec}}
\DeclareMathOperator\GL{GL}

\DeclareMathOperator\Staircase{Staircase}

\DeclareMathOperator\supp{supp}
\DeclareMathOperator\GCD{\textsc{gcd}}
\DeclareMathOperator\LCM{\textsc{lcm}}

\DeclareMathOperator\Card{\#}

\begin{document}

\begin{frontmatter}



\title{Guessing Gr\"obner Bases of Structured Ideals of Relations of Sequences}


\author{J\'er\'emy Berthomieu}
\cortext[cor1]{Laboratoire d'Informatique de Paris~6,
  Sorbonne Universit\'e, bo\^ite courrier~169, 4~place
  Jussieu, F-75252 Paris Cedex~05, France.}
\ead{jeremy.berthomieu@lip6.fr}

\author{Mohab Safey El Din}
\address{Sorbonne Universit\'e, \textsc{CNRS}, \textsc{LIP6},
  F-75005, Paris, France}
\ead{mohab.safey@lip6.fr}

\begin{abstract}
  Assuming sufficiently many terms of an $n$-dimensional table defined
over a field are given, we aim at guessing the linear recurrence
relations with either constant or polynomial coefficients they
satisfy.  In many applications, the table terms come along with a
structure: for instance, they may be zero outside of a cone, they may
be built from a Gr\"obner basis of an ideal invariant under the action
of a finite group. Thus, we show how to take advantage of this structure
to reduce both the number of table queries and the number of
operations in the base field to recover the ideal of relations of the
table. In applications like in combinatorics, where all these
zero terms make us guess many fake relations, this allows us to drastically
reduce these wrong guesses. These algorithms have been implemented
and, experimentally, they let us handle examples that we could not
manage otherwise.

Furthermore, we show which kind of cone and lattice structures are
preserved by skew-polynomial multiplication. This allows us to speed
the
guessing of linear recurrence relations with
polynomial coefficients up by computing
sparse Gr\"obner bases or Gr\"obner bases of an ideal invariant
under the action of a finite group in a ring of skew-polynomials.


\end{abstract}

\begin{keyword}
  Linear recurrence relations, Gr\"obner bases, Symmetries, Change of orderings
\end{keyword}

\end{frontmatter}


\section{Introduction}\label{s:intro}
\paragraph*{Problem statement and motivations}
Given a sequence
$\bseq=(\seq_{i_1,\ldots,i_n})_{i_1,\ldots,i_n\geq 0}$, we consider
the \emph{table} made of a finite subset of its terms.
Computing or \emph{guessing} linear recurrence relations satisfied by
such a table
is a fundamental problem in coding theory for cyclic
codes~\cite{BoseRC1960,Hocquenghem1959} of dimension $n\geq 1$,
combinatorics and computer algebra for solving sparse linear systems,
performing sparse polynomial interpolation,
polynomial least-square approximation and \gbs changes of
orderings in $n\geq 1$ variables~\cite{FaugereM2011,FaugereM2017}.
Furthermore, computing these
relations with polynomial coefficients in the indices allows us to
predict the growth of its terms, to classify the differential nature
of their generating series or to evaluate said generating
series~\cite{Mezzarobba2019}.

Depending on the context, an upper bound on the number of table terms
might be known in order to guess these relations. For instance, in coding
theory, this is related to the length and the minimum distance of the
code. In the \gbs change of orderings application, an upper bound is
given by the
degree of the ideal and the number of variables.
Whenever no upper bound is known, one is still
restricted to only consider a finite number of table terms to guess
the linear recurrence relations the table satisfies. Thus, some
of these relations may be proven incorrect when tested with many more
table terms; further such relations will be called \emph{fake
  relations}.
This happens for instance in combinatorics where the
nature itself of the table may be unknown.

In many applications, the table comes with a structure. For instance,
in combinatorics, for $n$D-space
walks in the nonnegative orthant, $\seq_{i_0,i_1,\ldots,i_n}$ counts
the number of ways to reach
$(i_1,\ldots,i_n)\in\N^n$ in $i_0$ steps of size
$1$~\cite{BostanBMKM2016,BousquetMP2003}.
Therefore, $\seq_{i_0,i_1,\ldots,i_n}$ is trivially $0$
outside the cone $i_1,\ldots,i_n\leq i_0$. Thus, computationwise, not
considering these terms would reduce the size of the table and thus
might be beneficial for guessing the
linear recurrence relations satisfied by the table.
Hence, the goal is
to exploit this structure to both reduce the number of table queries and the
number of operations to guess the \gb of the ideal of relations. 

\paragraph*{Prior results}
We distinguish two cases: the one-dimensional
case, where tables are with one index, and the multidimensional one,
where tables have $n>1$ indices.

In the one-dimensional case, given the first $D$
terms of a table, the \bmy
algorithm~\cite{Berlekamp1968,Massey1969} guesses the
linear recurrence relations with constant coefficients of smaller
order. Using 
fast extended Euclidean algorithm, this
algorithm can do so in $O(\sM(D)\log D)$
operations in the base field~\cite{BrentGY1980}, where
$\sM(D)=O(D\log D\log\log D)$~\cite{CantorK1991} is a cost
function for multiplying two univariate polynomials of degree at most $D$.
Through Hermite-Pad\'e approximants, the
\bln algorithm~\cite{BeckermannL1994} can be used to guess several
relations with polynomial coefficients including the one of minimal
order. Let us notice that finding relations with polynomial
coefficients is a special case of Hermite-Pad\'e
approximants for which the \bln algorithm is not quasi-optimal in the input
size.

In the multidimensional case, several algorithms were designed for
guessing linear recurrence relations with constant coefficients
satisfied by the first terms of the tables using linear algebra
routines. For instance, the \bms
algorithm~\cite{Sakata1988,Sakata1990,Sakata2009}, the \sFGLM
algorithm~\cite{BerthomieuBF2015,BerthomieuBF2017} or the \agbb
algorithm~\cite{Mourrain2017}. Given sufficiently many terms, the
first two return a \gb of the ideal of relations while the third one
returns a \bob of this ideal. Furthermore, in~\cite{BerthomieuF2018}
the authors designed an algorithm extending both the \bms and the \sFGLM
algorithms using polynomial arithmetic and in~\cite{BerthomieuF2016},
they extended the \sFGLM algorithm for guessing relations with
polynomial coefficients.
However, none of these algorithms were designed to
take the structure of the table terms into account. Another classical
technique is the ``ansatz + linear system solving'' approach for
finding relations. Usually the ansatz allows the user to find a set of
relations, then a post-processing is needed in order to recover, for
instance, the \gb of the ideal spanned by the computed
relations. Though, if the ansatz is far from being tight, then the
linear system to solve might be equivalent to the one of the \sFGLM algorithm.

\gbs are the output of several algorithms for guessing linear
recurrence relations and are a fundamental tool in polynomial systems
solving. In many applications, polynomials systems come with a
structure, for instance they span an ideal globally invariant under
the action of a finite group $G$ or their supports are in a cone. From
the table viewpoint, these are 
related to only considering table terms lying either on a
lattice~\cite{KauersV2019} or in a cone.

In~\cite{FaugereS2013}, the authors show that for such an ideal, \gbs
computations through the \Fquatre~\cite{Faugere1999},
\Fcinq~\cite{Faugere2002} and \FGLM~\cite{FaugereGLM1993} algorithms
can be sped up with a factor depending on $|G|$, whenever the
characteristic of the field of coefficients does not divide $|G|$. To
do so, they essentially perform $|G|$ parallel smaller
computations. In particular for the \FGLM algorithm, this factor is
$|G|^2$, see~\cite[Theorem~10]{FaugereS2013}. Likewise,
in~\cite{Steidel2013}, the author proposed algorithms for computing
\gbs of symmetric ideals over the rationals or a finite field.

In~\cite{BenderFT2018,FaugereSS2014}, the authors show that if $\cC$
is a semi-group of $\Z^n$ containing $0$ and no pair of opposite
elements and if $f_1,\ldots,f_s$ are polynomials with support in
the corresponding monomial set
$\cT[\cC]\coloneqq
\curl{x_1^{i_1}\cdots x_n^{i_n}\middle|\allowbreak
  (i_1,\ldots,i_n)\in\cC}$, then one can consider the ideal
spanned by
$f_1,\ldots,f_s$ in the subalgebra of
polynomials with support in $\cT[\cC]$. Modifying classical \gbs
algorithms, they obtain a \sgb, a set of generators with support in $\cT[\cC]$
of this ideal that behaves like a \gb. This allows them to
speed \gb computations up by taking into account the sparsity of the
union of the supports of the original generators of the ideal.

\paragraph*{Main results}
We design variants of the \sFGLM algorithm which guess linear
recurrence relations for an $n$-dimensional table $\bseq$, given as
polynomials in $x_1,\ldots,x_n$. The original algorithm is recalled
in \autopageref{algo:sFGLM}.

We first prove that restraining the \sFGLM algorithm to terms of a
table lying on a cone makes it compute a \sgb of the ideal of relations
of the table. More precisely, we obtain Theorem~\ref{th:cone_sFGLM}, a
simplified version of which is as follows.

\begin{theorem}
  Let $\cC$ be a semi-subgroup of $\N^n$ containing $0$. Let $\prec$ be a
  monomial ordering. Let
  $T\subset\cT[\cC]$ be a finite set of monomials ordered for $\prec$,
  such that for all $\mu_1,\mu_2\in\cT[\cC]$, if $\mu_1\mu_2\in T$,
  then $\mu_1$ and $\mu_2$ are in $T$.

  Let $\bseq$ be a $n$-dimensional table with nonzero elements
  $\seq_{i_1,\ldots,i_n}$ only if $(i_1,\ldots,i_n)\in\cC$.

  Then, if $T$ is large enough, the output of the \sFGLM called on
  $\bseq$, $T$ and $\prec$ is the reduced \sgb of the ideal of
  relations of $\bseq$ with support in $\cT[\cC]$.
\end{theorem}

Let us remark that this allows us to remove trivial constraints on the
relations
induced by the zero terms outside of the cone,
yielding in practice many fewer guessed relations
that eventually fail.  On the one hand, as a byproduct, this allows us to reduce the
number of table queries to guess the relations. For instance, for a
subtable of the
Gessel walk, using $3\,491$ table terms, we can guess $142$ relations
amongst which $136$ are fake and
only $6$ are correct.
On the other hand taking only
table terms in a cone allows us to consider table terms much further,
which in turn allow us to guess more relations. Indeed, with
$3\,010$ terms in a cone of the same table, we guess $21$ relations
and all of them are
correct. We refer to Table~\ref{tab:Prel} for more details. Let us
also notice that these fake relations may hide correct ones as their
leading monomials could divide the leading monomials of correct relations.

In the next theorem, we now  consider table terms lying on a
lattice $\Lambda$ and affine translates thereof. 
This allows us to design
a parallel variant of the \sFGLM algorithm, called the \lsFGLM
algorithm and given in \autopageref{algo:lattice_sFGLM}.
Assuming the fundamental domain of
$\Lambda$ has $L$ integer points, this variant essentially deals with
$L$ sets of table terms of sizes roughly divided by $L$. The following
theorem is a simplified version of Theorem~\ref{th:lattice_sFGLM}.

\begin{theorem}
  Let $\Lambda$ be a sublattice of $\Z^n$. Let $\prec$ be a monomial
  ordering. Let $T\subset\cT$ be a finite set of monomials ordered for
  $\prec$, such that for all $\mu_1,\mu_2\in\cT$, if $\mu_1\mu_2\in
  T$, then $\mu_1$ and $\mu_2$ are in $T$.

  Let $f_1,\ldots,f_s$ be polynomials spanning a zero-dimensional
  ideal $I$ of degree $D$ such that for all $1\leq i\leq
  s$, there exists $\ba\in\Z^n$ such that the support of $f_i$ is
  included in
  $\cT[\ba+\Lambda]\coloneqq\curl{x_1^{i_1}\cdots x_n^{i_n}\middle|\allowbreak
  (i_1,\ldots,i_n)\in\ba+\Lambda}$.
  
  Let $\bseq$ be a $n$-dimensional generic table whose ideal of relations is
  $I$.

  Then, if $T$ is large enough, then the output of the \lsFGLM called on
  $\bseq$, $T$ and $\prec$ is the reduced \gb of the ideal of
  relations of $\bseq$. Furthermore, each polynomial in this \gb has
  its support in a set $\cT[\ba+\Lambda]$.
\end{theorem}

Finally, we also make an adaptive variant of the \lsFGLM algorithm,
following what has been done for the \sFGLM. This adaptive variant,
given in \autopageref{algo:lattice_asFGLM},
aim
at reducing the number of table queries using the shape of the
staircase associated to the \gb of the ideal of relations of the table.

\paragraph*{Structure of the paper}
We first recall in Section~\ref{s:prelim} the classical connection
between linear recurrence relations with polynomial coefficients and
skew-polynomials in $2n$ variables. Then, we recall how using linear
algebra routines on a special kind of matrix, a \emph{multi-Hankel}
one, the \sFGLM algorithm, and
its adaptive variant the \asFGLM algorithm, guesses linear recurrence
relations.

In Section~\ref{s:guessing}, we design variants of the \sFGLM
algorithm that take the table structure into account for guessing
linear recurrence relations, then we prove
Theorems~\ref{th:cone_sFGLM} and~\ref{th:lattice_sFGLM}.
As an application, we provide a
modification of the \spFGLM algorithm~\cite{FaugereM2011,FaugereM2017}
whenever the ideal is globally invariant under the action of a finite group.

The same kind of variants of the \asFGLM algorithm are then designed,
in Section~\ref{s:adaptive}. Likewise, we prove
Theorem~\ref{th:lasFGLM} in this section. Then, we show how one
can perform skew-polynomial 
operations in order to preserve the cone and lattice structures of the
support of the polynomials.

Finally, in Section~\ref{s:bench}, we report on our speedup using our
\texttt{C} implementation of the \spFGLM algorithm when the ideal is
invariant under the action of a finite group. We also guess linear
recurrence relations satisfied by $n$D-space walks with and without 
exploiting the cone 
structure of the table and then test further the guessed relations. We
then report on how the cone structure allows us to guess fewer fake
linear recurrence relations.


\section{Preliminaries}\label{s:prelim}
\subsection{Tables and relations}
In all this paper, we take the convention that $0\in\N$. For $n\in\N$, $n\geq 1$, we let
$\bi=(i_1,\ldots,i_n)\in\N^n$, $\bx=(x_1,\ldots,x_n)$ and
$\bx^{\bi}=x_1^{i_1} \cdots x_n^{i_n}$.  For a subset $\cS$ of
$\N^n$, we let $\cT[\cS]=\curl{\bx^{\bs}\middle|\bs\in\cS}$ be the set
of monomials with exponents in $\cS$.  To ease the presentation, we
let $\cT\coloneqq\cT[\N^n]$. Finally, for a polynomial
$f=\sum_{\bs\in\cS}f_{\bs}\bx^{\bs}$, we let
$\supp f=\curl{\bs\in\cS\middle|f_{\bs}\neq 0}$ be its support.

Let $\K$ be a field
and $\bseq\in\K^{\N^n}$ be a $n$-indexed sequence with values
in $\K$, that is $\bseq=(\seq_{i_1,\ldots,i_n})_{(i_1,\ldots,i_n)\in\N^n}$.
There is a natural correspondence between finite linear combinations of terms of
$\bseq$ and polynomials in $\K[x_1,\ldots,x_n]$. For
$g=\sum_{\bs\in\cS}\gamma_{\bs}\bx^{\bs}$, with $\cS$ a finite subset of $\N^n$,
we can write $\brac{g}_{\bseq}\coloneqq\sum_{\bs\in\cS}\gamma_{\bs}\seq_{\bs}$.
Hence shifting a relation by an index $\bi$ comes down to
multiplying the corresponding polynomial by $\bx^{\bi}$ since
\[\brac{g\bx^{\bi}}_{\bseq}=\sum_{\bs\in\cS}\gamma_{\bs}\seq_{\bs+\bi}.\]

In particular, a polynomial $g$ defines a \emph{linear recurrence relation
with constant coefficients}, or \emph{\Crel} for short, on $\bseq$ if, and only if,
for all $\bi\in\N^n$,
$\brac{g\bx^{\bi}}_{\bseq}=0$.
The set of
all such polynomials is an ideal of $\K[\bx]$ called the \emph{ideal of \Crels
  of $\bseq$}, see for instance~\cite[Definition~2 and Proposition~4]{BerthomieuBF2017}.

Finally, a nonzero sequence $\bseq$ is said to be \emph{\Cfin} if
together with a finite number of terms of $\bseq$ and a finite number of
\Crels, one can recover all the terms of
$\bseq$. This is equivalent to requiring that
the ideal of \Crels of
$\bseq$ is $0$-dimensional, see also~\cite[Definition~2 and
Proposition~3]{BerthomieuBF2015}, where such sequences are called
\emph{linear recursive}.

\begin{example}\label{ex:init}
  On the one hand, the terms
  $v_{i,j}=(5+4i+3j) 2^{i+j}+(3+6i+j) 5^{i+j}$
  of $\bseq\in\F_7^{\N^2}$ 
  can all be computed
  thanks to $\seq_{0,0}=\seq_{0,1}=\seq_{0,2}=1$, $\seq_{1,0}=0$ and
  the \Crels, for all $(i,j)\in\N^2$,
  \[
    \seq_{i+1,j+1}+3\seq_{i,j}
    =\seq_{i+2,j}+\seq_{i,j+2}+6\seq_{i,j}
    =\seq_{i,j+3}+4\seq_{i+1,j}+6\seq_{i,j+1}=0.
  \]
  On the other hand, they can also be computed knowing
  $\seq_{0,0}=\seq_{0,1}=\seq_{0,2}\allowbreak=\seq_{0,3}=1$ and that
  for all $(i,j)\in\N^2$,
  \[
  \seq_{i,j+4}+6\seq_{i,j+2}+2\seq_{i,j}
  =\seq_{i+1,j}+2\seq_{i,j+3}+5\seq_{i,j+1}=0.
  \]
  Thus, the ideal of \Crels of $\bseq$ is the $0$-dimensional one
  $\ideal{x y+3,x^2+y^2+6,y^3+4x+6y}
  =\ideal{y^4+6y^2+2,x+2y^3+5y}$ and $\bseq$ is \Cfin.

  On the other hand, the binomial sequence,
  $\bbin=\pare{\bin_{i,j}}_{(i,j)\in\N^2}=\pare{\binom{i}{j}}_{(i,j)\in\N^2}$,
  satisfies Pascal's rule: for all $(i,j)\in\N^2$,
  $\bin_{i+1,j+1}-\bin_{i,j+1}-\bin_{i,j}=0$. Moreover, one can 
  show that this relation spans all the other \Crels, \ie
  its ideal of \Crels is the $1$-dimensional one $\ideal{x y-y-1}$,
  thus $\bbin$ is not \Cfin.
\end{example}

Furthermore, some sequences satisfy \emph{linear recurrence relations with
coefficients that are polynomials} in the indices of the sequence, or
\emph{\Prels} for short. For
instance, the binomial sequence satisfies the following two \Prels for
all $(i,j)\in\N^2$: 
\begin{align*}
  \phantom{i-}\pare{j+1}\bin_{i,j+1}-\pare{i-j}\bin_{i,j}&=0\\
  \pare{i+1-j}\bin_{i+1,j}-\pare{i+1}\bin_{i,j}&=0.
\end{align*}
Combining them by shifting the former by index $(0,1)$ and then
adding the latter yields
\[(i-j)\bin_{i+1,j+1}-(i-j)\bin_{i,j+1}-(i-j)\bin_{i,j}=0.\]
This
proves that Pascal's rule holds whenever $i\neq j$.

We thus aim at representing the former relations as polynomials $g_1$
and $g_2$ such that for all $(i,j)\in\N^2$,
$\brac{g_1x^i y^j}_{\bseq}=\brac{g_2x^i y^j}_{\bseq}=0$.  For instance, we could say
that the first one corresponds to
$\brac{(j+1)x^i y^{j+1}-(i-j)x^i y^j}_{\bseq}=\brac{\pare{(j+1)y-(i-j)}x^i y^j}_{\bseq}=0$,
but this would mean that $g_1$ has coefficients in $i$ and $j$, which
are meaningless on their own.  To circumvent this,
in~\cite{BerthomieuF2016}, the authors introduced new
variables $\bt=(t_1,\ldots,t_n)$, such that $t_p$ behaves like
$x_p\partial_p$, where $\partial_p$ is the differential operator with
respect to $x_p$. That is,
$\brac{\bt^{\bk}\bx^{\bi}}_{\bseq}\coloneqq
\brac{t_1^{k_1}\cdots t_n^{k_n}\bx^{\bi}}_{\bseq}
=\brac{(x_1\partial)^{k_1}\cdots(x_n\partial_n)^{k_n}\bx^{\bi}}_{\bseq}
=\brac{i_1^{k_1}\cdots i_n^{k_n}\bx^{\bi}}_{\bseq}=i_1^{k_1}\cdots
i_n^{k_n}\seq_{\bi}=\bi^{\bk}\seq_{\bi}$. Then, the $\brac{.}_{\bseq}$
notation is naturally $\K$-linearly extended to polynomials in $\bt$
and $\bx$.  Therefore, the $2n$ variables
$t_1,\ldots,t_n,x_1,\ldots,x_n$ follow, for all $1\leq p,q\leq n$
and $p\neq q$, the commutation rules $x_p x_q=x_q x_p$,
$t_p t_q=t_q t_p$, $t_p x_q=x_q t_p$ and
$t_p x_p=x_p(t_p+1)$, making polynomials in $\bt$ and $\bx$
\emph{quasi-commutative}. The ring of skew-polynomials in $\bt$ and
$\bx$, satisfying the quasi-commutative rules defined above, will be
denoted $\K[\bt]\ideal{\bx}$ while the ring of
skew-polynomials in $\bx$ with coefficients in $\K(\bt)$ will simply
be
denoted $\K(\bt)\ideal{\bx}$. Now, a \Prel is given by a finite subset
$\cS$ of $\N^n$ and polynomials $\gamma_{\bs}\in\K[\bt]$ for
$\bs\in\cS$, such that
\[\forall\bi\in\N^n,
\sum_{\bs\in\cS}\gamma_{\bs}(\bs+\bi)\seq_{\bs+\bi}=0.\]
This relation corresponds to the polynomial
$g=\sum_{\bs\in\cS}\gamma_{\bs}(\bt)\bx^{\bs}\in\K[\bt]\ideal{\bx}$ such that for
all $\bi\in\N^n$, $\brac{g\bx^{\bi}}_{\bseq}=0$.

\begin{remark}
  While we can obviously find polynomials
  $\tilde{\gamma}_{\bs}\in\K[\bt]$ such that
  $\sum_{\bs\in\cS}\tilde{\gamma}_{\bs}(\bi)\seq_{\bs+\bi}
  =\sum_{\bs\in\cS}\gamma_{\bs}(\bs+\bi)\seq_{\bs+\bi}
  =\brac{g\bx^{\bi}}_{\bseq}=0$,
  the notation with the
  $\gamma_{\bs}(\bs+\bi)$'s makes more explicit the
  corresponding polynomial $g$ in $\K[\bt]\ideal{\bx}$.
\end{remark}
\begin{example}\label{ex:binom_ideal}
  Let $t=t_1$, $u=t_2$, $x=x_1$ and $y=x_2$. Then, the \Prels
  satisfied by the binomial sequence can be rewritten as
  \begin{align*}
    \phantom{i-}\pare{j+1}\bin_{i,j+1}-\pare{i-j}\bin_{i,j}
    &=\brac{\pare{j+1}x^iy^{j+1}-\pare{i-j}x^iy^j}_{\bseq}\\
    0 &=\brac{ux^iy^{j+1}-\pare{t-u}x^iy^j}_{\bseq}\\
    0 &=\brac{\pare{u y-\pare{t-u}}x^iy^j}_{\bseq}\\
    \text{and\qquad}
    \pare{i+1-j}\bin_{i+1,j}-\pare{i+1}\bin_{i,j}
    &=\brac{\pare{i+1-j}x^{i+1}y^j-\pare{i+1}x^iy^j}_{\bseq}\\
    0 &=\brac{\pare{t-u}x^{i+1}y^j-\pare{t+1}x^iy^j}_{\bseq}\\
    0 &=\brac{\pare{\pare{t-u}x-\pare{t+1}}x^iy^j}_{\bseq}.
  \end{align*}
  Thus, $g_1=u y-\pare{t-u}$ and $g_2=\pare{t-u}x-\pare{t+1}$ in $\K[t,u]\ideal{x,y}$.
\end{example}
The set of all such
polynomials is a right ideal of $\K[\bt]\ideal{\bx}$. Indeed, it is
stable by multiplication on the right by any monomial $\bx^{\bi}$ as
requested. Furthermore, since
$\bt^{\bl}\bx^{\bj}\bt^{\bk}\bx^{\bi}
=\bt^{\bl}\pare{\bt-\bj}^{\bk}\bx^{\bj+\bi}$,
$\brac{\bt^{\bl}\bx^{\bj}\bt^{\bk}\bx^{\bi}}_{\bseq}
=\brac{\bt^{\bl}\pare{\bt-\bj}^{\bk}\bx^{\bj+\bi}}_{\bseq}
=\pare{\bj+\bi}^{\bl}\bi^{\bk}\seq_{\bj+\bi}
=\bi^{\bk}\brac{\bt^{\bl}\bx^{\bj+\bi}}_{\bseq}$. In other words,
multiplying on the right by $\bt^{\bk}\bx^{\bi}$ corresponds to
multiplying on the right by $\bx^{\bi}$ and to multiply the evaluation by a
constant, namely $\bi^{\bk}$. Thus if $\brac{g\bx^{\bi}}_{\bseq}$ vanishes, then so does
$\brac{g\bt^{\bk}\bx^{\bi}}_{\bseq}$.

Such relations allow one to compute new terms of the sequence, though
integer roots of the leading coefficient may prevent some
computations. For instance, one cannot compute $\bin_{i+1,i+1}$ from
$\bin_{i,i+1}$ using $(i+1-j) \bin_{i+1,j}-(i+1)\bin_{i,j}$ and
$j=i+1$ as the
coefficient in front of $\bin_{i+1,j}$ vanishes. Thankfully, for this
sequence, one can use the other relation
$(j+1)\bin_{i,j+1}-(i-j)\bin_{i,j}$ with $i=j+1$ to achieve this goal.

Sequences satisfying \Prels form a large set. Among them, there are
the \emph{\Pfin} ones. In particular, analogously to the \Cfin case, a
nonzero sequence $\bseq$ such that
a finite number of its terms and a finite number of
\Prels allows one to recover all of its terms is \Pfin.
\begin{example}[Cont.\ of Example~\ref{ex:binom_ideal}]
  The ideal of \Prels of~$\bbin$ in $\K[t,u]\ideal{x,y}$ is
  \[\ideal{u y-(t-u),(t-u)x-(t+1),x y-y-1}.\]
  Furthermore, since
  \begin{align*}
    \pare{xy-y-1}\pare{t-u}
    &= \pare{t-u}xy - \pare{t+1-u}y - \pare{t-u}\\
    &= \pare{\pare{t-u}x - \pare{t+1}}y + \pare{u y - \pare{t-u}},
  \end{align*}
  in $\K(t,u)\ideal{x,y}$, its ideal of \Prels is only spanned by
  $u y-(t-u)$ and $(t-u)x-(t+1)$.
\end{example}
Note that \Pfin sequences are actually those whose generating series
are D-finite and there
exist \Pfin sequences that do not satisfy the above prerequisites.

\subsection{\gbs}
This section briefly recalls some basic definitions on Gr\"obner
bases. The interested reader will find more details
in~\cite{CoxLOS2015} in the commutative case
and~\cite[Chapter~2]{Levandovskyy2005} in the quasi-commutative one.

For $\cT$ the set of monomials in $\K[\bt]\ideal{\bx}$, a monomial ordering
$\prec$ on $\cT$ is a total order relation satisfying the following
three properties
\begin{enumerate}
\item $\forall m\in\cT$, $1\preceq m$;
\item $\forall m,m',s\in\cT$, $m\preceq m'\Rightarrow m s\preceq
  m' s$ and $s m\preceq s m'$.
\end{enumerate}

For a monomial ordering $\prec$ on $\K[\bt]\ideal{\bx}$, the \emph{leading
  monomial} of $f$, denoted $\LM(f)$, is the greatest monomial in the
support of $f$ for $\prec$. 
For an ideal $I$, we let
$\LM(I)=\{\LM(f),\ f\in I\}$.
We recall briefly the definition of a \gb and of its associated staircase.
\begin{definition}\label{def:staircase}
  Let $I$ be a nonzero ideal of $\K[\bt]\ideal{\bx}$ and let $\prec$ be
  a monomial ordering.
  A set $\cG\subseteq I$ is a \emph{\gb} of $I$ if for all $f\in I$,
  there exists $g\in\cG$ such that $\LM(g)|\LM(f)$,
  it is \emph{reduced} if for any $g,g'\in\cG$,
  and $g\neq g'$, any monomial $m\in\supp g'$ satisfies $\LM(g)\nmid m$.

  The \emph{staircase of $\cG$}
  is defined as
  $S=\Staircase(\cG)=\{s\in\cT,\ \forall g\in\cG, \LM(g)\nmid s\}$.

  More generally, a set $S$ will be said to be a \emph{staircase} if
  for two monomials $\mu_1$ and $\mu_2$ such that $\mu_1\mu_2\in S$,
  we have
  $\mu_1\in S$ and $\mu_2\in S$.
\end{definition}
Let us recall that $\Staircase(\cG)$ is also the canonical basis of
$\K[\bt]\ideal{\bx}/I$ as a $\K$-vector space.

\gb theory allows us to choose any monomial ordering,
among which we mainly use, on the $\bx$ variables, the
\begin{description}
\item[$\LEX(x_n\prec\cdots\prec x_1)$ ordering] which satisfies
  $\bx^{\bi}\prec\bx^{\bj}$ if, and only if, there
  exists $1\leq p\leq n$ such that for all $q<p$,
  $i_q=j_q$ and $i_p<j_p$,
  see~\cite[Chapter~2, Definition~3]{CoxLOS2015};
\item[$\DRL(x_n\prec\cdots\prec x_1)$ ordering] which satisfies
  $\bx^{\bi}\prec\bx^{\bj}$ if, and only if,
  $i_1+\cdots+i_n<j_1+\cdots+j_n$ or $i_1+\cdots+i_n=j_1+\cdots+j_n$
  and there exists
  $2\leq p\leq n$ such that for all $q>p$, $i_q=j_q$
  and $i_p>j_p$,
  see~\cite[Chapter~2, Definition~6]{CoxLOS2015}.
\end{description}
We will also use monomial orderings on the $\bt$ and $\bx$
variables. Since we want to freely switch from
$\K[\bt]\ideal{\bx}$ to $\K(\bt)\ideal{\bx}$ and vice versa, it makes
sense to choose an ordering such that $t_k\prec x_{\ell}$ for any $k$
and $\ell$, such as
$\LEX(t_n\prec\cdots\prec t_1\prec x_n\prec\cdots x_1)$ or
$\DRL(t_n\prec\cdots\prec t_1\prec\cdots\prec x_1)$. The latter is
more suitable as it allows us to enumerate all the monomials in $\bt$
and $\bx$ in increasing order.

\subsection{Structured \gbs}
The \emph{cones} we are
dealing with are those that are \emph{submonoids} of $\N^n$. These are
subsets $\cC$ of $\N^n$ such that
$0\in\cC$ and for all $\bi,\bj\in\cC$,
$\pare{\bi+\bj}\in\cC$.

Given such a cone $\cC$ and polynomials with support in its
associated set of monomials
$\cT[\cC]=\curl{\bx^{\bi}\in\cT\middle| \bi\in\cC}$, one may want to perform all
the polynomial operations with monomials in $\cT[\cC]$ in order to take advantage of
the structure of the support when computing a \gb of the ideal they
span. While, this is not always possible, one can achieve this goal by
considering the ideal the polynomials span in the subalgebra defined
by $\cC$.

This leads to the definition of \sgb with support in
$\cT[\cC]$ that uses its monoid structure.
\begin{definition}[{\cite[Definition~3.1]{FaugereSS2014}}
  and {\cite[Definition~3.3]{BenderFT2018}}]
  Let $\cC\subseteq\N^n$ be a cone and $\cT[\cC]$ be its associated
  set of monomials. Then, $\K[\cC]$, the set of polynomials with
  support in $\cT[\cC]$, is an algebra.

  Let $f_1,\ldots,f_s\subseteq\K[\cC]$ be
  polynomials. We let
  $I=\ideal{f_1,\ldots,f_s}_{\cC}=\curl{\sum_{k=1}^s f_k
    q_k\middle|q_1,\ldots,q_s\in\K[\cC]}$ be the 
  \emph{ideal spanned by $f_1,\ldots,f_s$ in $\K[\cC]$}.
  Then, a \emph{\sgb} of $I$ for a monomial ordering
  $\prec$ is a generating set $\cG=\curl{g_1,\ldots,g_r}\subseteq\K[\cC]$ such that
  for all
  $f\in I$,
  $\LM(f)=\LM(g)m$ for some $g\in\cG$ and
  $m\in\cT[\cC]$.

  The associated staircase $\Staircase(\cG)$ of $\cG$ is the set of
  monomials $s$ in $\cT[\cC]$ such that for any $g\in\cG$, there is no
  monomial $m\in\cT[\cC]$ such that $s=\LM(g)m$.
\end{definition}
Let us notice that for $\cC=\N^n$, $\K[\cC]=\K[\N^n]=\K[\bx]$ and \sgbs are classical
\gbs. Furthermore, like classical \gbs, \sgbs allow one to solve the
ideal membership problem in $\K[\cC]$ in an effective way.

For a \emph{lattice} $\Lambda\subseteq\Z^n$, we let $\Lambda_{\pos}=\Lambda\cap\N^n$
be its nonnegative cone, so that, naturally, $\Z^n_{\pos}=\N^n$. In
particular, $\Lambda$ and $\Lambda_{\pos}$ are cones and we
may intersect them with another cone. For
$\ba\in\Z^n$, we also denote by $\ba+\Lambda$ the affine lattice obtained by
translating $\Lambda$ by $\ba$ and likewise we can consider its intersection with a
cone. In particular, $(\ba+\Lambda)_{\pos}=(\ba+\Lambda)\cap\N^n$.

Given a lattice
$\Lambda$, its affine translates
$\ba_0+\Lambda=\Lambda,\ldots,\ba_L+\Lambda$ and polynomials
$f_1,\ldots,f_k$, each with 
supports in an associated set of monomials
$\cT[(\ba_{\ell}+\Lambda)_{\pos}]$, then a reduced \gb of
$\ideal{f_1,\ldots,f_k}$ satisfies also this support property. This
allows one to speed the \gbs computations up by essentially performing
$L$ computations in parallel with input of sizes divided by $L$.

\subsection{Multi-Hankel matrices}
Given a table $\bseq$ and a polynomial $g\in\K[\bt]\ideal{\bx}$, in
order to determine if a polynomial $g$ is in the ideal of \Prels of $\bseq$, one
must check that $\brac{g\bx^{\bi}}_{\bseq}=0$ for all $\bi$. As only a finite
number of terms of $\bseq$ are known, only a finite number of such
tests can be done.

\begin{definition}
  Let $T$ be a finite subset of $\cT[\N^{2 n}]$, the set of monomials
  in $t_1,\ldots,t_n,x_1,\ldots,x_n$, and $X$ be a finite subset of
  $\cT[\N^n]$, the set of monomials in $x_1,\ldots,x_n$.

  The \emph{multi-Hankel matrix $H_{X,T}$} is the matrix whose rows are
  indexed by $X$ and columns by $T$ and whose coefficient at row
  $\bx^{\bi}$ and column $\bt^{\bk}\bx^{\bj}$ is $\brac{\bt^{\bk}\bx^{\bj+\bi}}_{\bseq}$.
\end{definition}
A vector in the right kernel of this matrix corresponds to a polynomial
$g$ with support in $T$ such that $\brac{g m}_{\bseq}=0$ for all $m\in
X$.

\begin{example}
  Let $\bseq=(\seq_{i,j})_{(i,j)\in\N^2}$ be a table and
  $T=\curl{1,u,t,y,x,u y,t y,u x,t x}\subset\cT[\N^{2n}]$ and
  $X=\curl{1,y,x,y^2,x y,x^2}\subset\cT[\N^n]$ be two sets of monomials,
  then their multi-Hankel matrix is
  \[H_{X,T}=
    \kbordermatrix{
      \ 		&1	
      &u		&t		&y		&x
      &u y		&t y		&u x		&t x\\
      1		&\seq_{0,0}
      &0		&0		&\seq_{0,1}	&\seq_{1,0}
      &\seq_{0,1}	&0		&0		&\seq_{0,1}\\
      y		&\seq_{0,1}
      &\seq_{0,1}	&0		&\seq_{0,2}	&\seq_{1,1}
      &2\seq_{0,2}	&0	&\seq_{1,1}	&\seq_{1,1}\\
      x		&\seq_{1,0}
      &0		&\seq_{0,1}	&\seq_{1,1}	&\seq_{1,1}
      &\seq_{1,1}	&\seq_{1,1}	&0		&2\seq_{2,0}\\
      y^2	&\seq_{0,2}
      &2\seq_{0,2}	&0		&\seq_{0,3}	&\seq_{1,2}
      &3\seq_{0,3}	&0	&2\seq_{1,2}	&\seq_{1,2}\\
      x y	&\seq_{1,1}
      &\seq_{1,1}	&\seq_{1,1}	&\seq_{1,2}	&\seq_{2,1}
      &2\seq_{1,2}	&\seq_{1,2}	&\seq_{2,1}	&2\seq_{2,1}\\
      x^2		&\seq_{2,0}
      &0		&2\seq_{2,0}	&\seq_{2,1}	&\seq_{3,0}
      &\seq_{2,1}	&2\seq_{2,1}	&0		&3\seq_{3,0}\\
    }.
  \]
  We give some computation details. The coefficient on the third
  column ($t$) and first row ($1$) is
  $\brac{t\times 1}_{\bseq}=\brac{t x^0 y^0}_{\bseq}=0^1\seq_{0,0}=0$. Likewise, the
  coefficient on sixth column ($u y$) and the second to last row ($x y$)
  is $\brac{u y x y}_{\bseq}=\brac{u x y^2}_{\bseq}=2^1\seq_{1,2}=2\seq_{1,2}$.

  Note that rows are only indexed with monomials in $\bx$ and not in
  $\bt,\bx$ since the row labeled with $\bt^{\bk}\bx^{\bi}$,
  $\bk\neq 0$ would
  be a multiple of the row labeled with $\bx^{\bi}$.
\end{example}

\subsection{The \sFGLM algorithm}
The \sFGLM algorithm~\cite{BerthomieuBF2015,BerthomieuBF2017}, takes
as input the table $\bseq$ and a set of monomials $T$, which is a staircase,
and computes the right kernel of the multi-Hankel matrix
$H_{T,T}$. Vectors in this kernel can be seen as polynomials in
$\K[\bx]$ and these polynomials
with a leading term minimal for the partial order induced by the
division are the ones returned by the algorithm.
Furthermore, if $T$ is ordered for a monomial ordering $\prec$ and contains
the staircase and the leading monomials of the reduced \gb of the
ideal of \Crels of $\bseq$ for $\prec$, then the \sFGLM algorithm
returns this \gb.

As our goal is to extend the \sFGLM algorithm in order to deal with
table terms lying on a cone or a lattice, we recall this algorithm.
\begin{algorithm2e}[htbp!]
  \small
  \DontPrintSemicolon
  \TitleOfAlgo{\sFGLM\label{algo:sFGLM}}
  \KwIn{A table $\bseq=(\seq_{\bi})_{\bi\in\N^n}$ with coefficients in
    $\K$, a monomial ordering $\prec$, a sufficiently large staircase
    $T$
    and ordered
    for $\prec$.}
  \KwOut{A reduced \gb of the ideal of \Crels of $\bseq$.}
  Build the matrix $H_{T,T}$.\;
  Compute the set $S\subseteq T$ of smallest monomials, for $\prec$,
  such that $\rank H_{S,S}=\rank H_{T,T}$.\;
  \Forall(\tcp*[f]{stabilize $S$ for the division}){$m\in T\setminus S$}{
    \lIf{$\exists s\in S$ such that $m\mid s$}{
      $S\coloneqq S\cup\{m\}$.
    }
  }
  $L\coloneqq T\setminus S$ sorted for $\prec$.\;
  $G\coloneqq\varnothing$.\;
  \While{$L\neq\varnothing$}{
    $g\coloneqq\min_{\prec} L$\;
    Solve the linear system
    $H_{S,S}\bgamma+H_{S,\curl{g}}=0$.\;
    $G\coloneqq G\cup\curl{g+\sum_{s\in S}\gamma_s s}$.\;
    Remove $g$ and any of its multiples from $L$.
  }
  \KwRet $G$.
\end{algorithm2e}

The algorithm computes the column rank profile of the matrix
$H_{T,T}$, that is the set of leftmost linearly independent columns of
the matrix. Since these columns are independent from the previous
ones, their labels cannot be the leading monomial, for $\prec$,
of any polynomial in the ideal of \Crels, thus they are in the
associated staircase of the reduced \gb of this ideal for $\prec$. If
$T$ is not large enough, a monomial $m$ could be detected as not lying
in the staircase while one of its multiples does, hence there is a
stabilization process to add $m$ to the staircase if this
happens, see~\cite[Example~3]{BerthomieuBF2015}. Then, each output
polynomial is computed by solving a linear system involving its
leading monomial and the monomials in the staircase.

\begin{example}[Cont.\ of Example~\ref{ex:init}]\label{ex:sFGLM}
  Let us recall that a \gb of the
  ideal of \Crel of $\bseq$ is
  $\curl{x y+3,x^2+y^2+6,\allowbreak y^3+4x+6y}$
  for $\DRL(y\prec x)$,
  hence this ideal has degree~$4$.
  Therefore, the staircase of the \gb of this ideal for
  $\LEX (y\prec x)$, or any monomial ordering, can only contain monomials
  $x^i y^j$ with $(i+1)\times (j+1)\leq 4$ and it suffices to take
  $T=\curl{1,y,y^2,y^3,y^4,x,x y,x^2,x^3,x^4}$ to recover the staircase
  and the \gb. 
  The column rank profile of $H_{T,T}$ is given by
  $S=\curl{1,y,y^2,y^3}$ so that
  $L=\curl{y^4,x,x y,x^2,x^3,x^4}$. Then, the linear systems
  $H_{S,S}\bgamma+H_{S,\{y^4\}}=0$ and
  $H_{S,S}\bgamma+H_{S,\{x\}}=0$ yield the \gb
  \[\curl{y^4+6y^2+2,x+2y^3+5y}.\]
\end{example}
In many applications, for instance the \gbs change of orderings one
through the \spFGLM algorithm, the
computation of a single table element is costly. Therefore, we may
want to reduce the number of table queries performed by the \sFGLM
algorithm. 
Algorithm~\ref{algo:sFGLM} called on a set $T$ requires $\Card{2T}$ table terms, where
$2T$ is the Minkowski sum of $T$ with itself. To reduce this number of
queries, the
goal is to let the multi-Hankel grow step by step. We start with the
$1\times 1$ matrix
\[\kbordermatrix{
    \ 	&1\\
    1	&\brac{1}_{\bseq}
  }.\]
If $\brac{1}_{\bseq}=\seq_{\bzero}\neq 0$, then $1$ is in the associated
staircase of the \gb of the ideal of \Crels of $\bseq$, otherwise it
stops and returns the set of relations $\curl{1}$. The algorithm extends a full-rank
matrix $H_{S,S}$ into $H_{S\cup\{m\},S\cup\{m\}}$ with $m$ greater,
for $\prec$, than any monomial in $S$. Now, there are two
possibilities, either the new matrix has full rank or it is not and the
column labeled with $m$ is linearly dependent from the other ones. In
the former case, $m$ is actually in this staircase and $S$ is replaced
by $S\cup\{m\}$. In the latter case, a polynomial with support in
$S\cup\{m\}$ and leading monomial, for $\prec$, $m$ is found and no
multiples of $m$ will ever be proposed to extend the multi-Hankel
matrix. The algorithm stops either when no monomials can be added to
the staircase or when the size of the staircase has reached a threshold
given in input. There is, however, a possibility of finding wrong
relations if the first terms of the table exceptionally satisfies a
relation of smaller order, for instance if $\seq_{\bzero}=0$. This
problem can be circumvented by testing relations further, that is
adding a small buffer of constraints, \ie rows of the matrix. This can
be noticed for instance when the
relations are suspiciously small or in \FGLM
applications where the degree of the ideal is known in advance.


\section{Guessing with structures}\label{s:guessing}
In this section, we show how to guess linear recurrence relations of a
table by taking the structure of the table terms into account. We
first start with the case where only table terms in a cone are
considered. Then, we study how to guess these relations when table
terms are in a lattice or some affine translates thereof.

\subsection{Terms in a cone}\label{ss:guessing_cone}
In this subsection, we aim at describing how we can take advantage of
the structure of a given cone $\cC$ to recover the ideal of relations
of a table $\bseq$ by only considering table terms inside the cone. That is,
we aim at guessing polynomials $g\in\K[\cC]$ such that
for all
$\bx^{\bi}\in\cT[\cC]$, $\brac{g\bx^{\bi}}_{\bseq}=0$. This latter
condition is the guessing part as we will only be able to ensure that
$\brac{g\bx^{\bi}}_{\bseq}=0$ for all $\bx^{\bi}$ in a finite subset $T$ of $\cT[\cC]$.

To do so, two strategies are at our disposal and they both rely on the
generators of $\cC$ as a submonoid of $\N^n$. Let us denote by
$\ba_1,\ldots,\ba_{\nu}$ a set of generators of $\cC$, \ie for all
$\bi\in\cC$, there exists $\bj\in\N^{\nu}$ such that
$\bi=j_1\ba_1+\cdots+j_{\nu}\ba_{\nu}$. Then, note that, first and foremost, there is
no reason for $\nu$ to be less than or equal to $n$. Second, even if
$\nu$ is minimal and
$\ba_1,\ldots,\ba_{\nu}$ is a generating set, there is no reason for
$\pare{j_1,\ldots,j_{\nu}}$ to be unique.
\begin{example}\label{ex:cone}
  The cone $\cC=\curl{\bi\in\N^2\middle|
    i_1\leq 2i_2,i_2\leq 2i_1}$ is spanned by $\ba_1=(1,1)$,
  $\ba_2=(1,2)$ and $\ba_3=(2,1)$ so that
  $\cC=\curl{j_1\ba_1+j_2\ba_2+j_3\ba_3\middle|
    \ba_1
    =(1,1),\ba_2=(1,2),\ba_3=(2,1),(j_1,j_2,j_3)\in\N^3}$. Yet,
  we have the two decompositions $(3,3)=3\ba_1=\ba_2+\ba_3$.
\end{example}

The first strategy is designed to only consider table terms
lying in $\cC$. Assuming a generating set $\ba_1,\ldots,\ba_{\nu}$ of
$\cC$ is known, the set of monomials $\cT[\cC]$ can be defined as
\[\cT[\cC]=\curl{\bx^{j_1\ba_1}\cdots\bx^{j_{\nu}\ba_{\nu}}\middle|
    (j_1,\ldots,j_{\nu})\in\N^{\nu}}.\] The second strategy makes use
of a new set of variables $\by=(y_1,\ldots,y_{\nu})$, so that $y_1$
represents $\bx^{\ba_1}$, etc and an auxiliary table
$\bseqq=(\seqq_{\bj})_{\bj\in\N^{\nu}}$ defined by
$\seqq_{\bj}=\seq_{j_1\ba_1+\cdots+j_{\nu}\ba_{\nu}}$. Then, two
monomials $\by^{\bj}$ and $\by^{\bk}$ represent the same monomial
$\bx^{\bi}$ if, and only if,
$\bi=j_1\ba_1+\cdots+j_{\nu}\ba_{\nu}=k_1\ba_1+\cdots+k_{\nu}\ba_{\nu}$. This
implies that
both $\seqq_{\bj}$ and $\seqq_{\bk}$ are equal to
$\seq_{\bi}$. Thus, $\bseqq$ satisfies extra relations coming from
these multiple equivalent writings. They are given by binomials,
namely $\by^{\bj}-\by^{\bk}$. Hence, not all monomials in
$\cT[\N^{\nu}]$ are of interest and we clean them up by using the
binomial ideal $I(\cC)$ they span, for instance by reducing $\by^{\bj}$ to $\by^{\bk}$.

In practice, both strategies are equivalent. They only differ in how
they enumerate table terms $\seq_{\bi}$ with $\bi\in\cC$. Note,
though, that the second strategy requires
computing a \gb of $I(\cC)$, for instance
using~\cite{KoppenhagenM1999} while the first one only requires
checking that a monomial has already been generated. However, such a
\gb computation should not be the bottleneck compared to the
computations of the table terms or the linear algebra routines for the
guessing step.

Since the first strategy comes down to directly calling the \sFGLM algorithm
with a set of monomials 
$T\subset\cT[\cC]$, this yields Theorem~\ref{th:cone_sFGLM}.
\begin{theorem}\label{th:cone_sFGLM}
  Let $\cC$ be a submonoid cone of $\N^n$ spanned by the minimal set
  of generators $\curl{\ba_1,\ldots,\ba_{\nu}}$. Let $\prec$ be a
  monomial ordering on $\cT$, the set of monomials in $n$ variables,
  and let $T\subset\cT[\cC]$ be a staircase
  ordered for~$\prec$.

  Then, the \sFGLM algorithm called on table $\bseq$, $T$ and $\prec$ returns a
  set of polynomials $G$ with support in $\cT[\cC]$, such that for all
  $s\in T\setminus\ideal{\LM(G)}$, $s$ is in the associated
  staircase of a \sgb of the ideal of \Crels of $\bseq$ for $\prec$.

  Furthermore, if the ideal of \Crels of $\bseq$ is $0$-dimensional
  and has a reduced \sgb
  with support in $T$ for $\prec$, then the output of the \sFGLM algorithm called on
  $\bseq$ and $T$ is this reduced \sgb.
\end{theorem}
\begin{proof}
  As the \sFGLM algorithm computes kernel vectors of $H_{T,T}$, the
  corresponding polynomials can only have support in $\cT[\cC]$.
  
  Let $S$ be the associated staircase of a \sgb of the ideal of \Crels
  of $\bseq$.

  Let us show first that no monomial $m\not\in S$ is found in the
  staircase by the algorithm. As $m\in\LM(I)$, there exist $\alpha_s\in\K$,
  for all $s\in S$ such that
  $m+\sum_{s\in S}\alpha_s s\in I$, thus
  $\brac{t\pare{m+\sum_{s\in S}\alpha_s s}}_{\bseq}=0$ for all
  $t\in\cT$. Since $T\subset\cT$, this means that the column labeled
  with $m$ is linearly dependent from the previous ones and neither
  $m$ nor any multiple thereof is in the staircase associated to the
  output.
  Hence, the
  computed staircase is included in the correct staircase.

  Let us now assume that the ideal of \Crels of $\bseq$ is
  $0$-dimensional, that is $S$ is finite. We shall show by
  contradiction that the matrix $H_{S,S}$ has full rank, so that the output
  of the \sFGLM algorithm called on $T\supset S$ is a reduced \gb
  whose associated staircase contains $S$. Let us assume that
  $H_{S,S}$ has not full rank and let $m\not\in S$ be the smallest
  monomial for $\prec$ such that
  $\rank H_{S,S\cup\curl{m}}>\rank H_{S,S}$, such a monomial exists
  for otherwise a monomial in $S$ would be the leading monomial of a
  relation.
  Let $R$ be any finite
  subset of $S\cup\curl{\mu\middle|\mu\preceq m}$, which is also a
  staircase
  containing $S\cup\{m\}$. By minimality of $m$, for $\prec$,
  $\rank H_{S,R}=\rank H_{S,S\cup\curl{m}}>\rank H_{S,S}$ and in
  particular the column labeled with $m$ must be independent from the
  previous ones. Thus, no polynomial with leading monomial $m$ can be
  in the ideal of relations and $m$ is in the staircase of this
  ideal. This is a contradiction with the assumption that $m$ is not
  in $S$. Since $S\subseteq T$, the algorithm correctly computes
  a superset of the staircase $S$ and thus the algorithm discovers the
  correct staircase.

  Finally, the polynomials of the \sgb are found by
  linear algebra.
\end{proof}

Concerning the second strategy,
since $I(\cC)$ is spanned by
binomials, the reduced \gb $\cG$ of $I(\cC)$ for $\prec$ is made of binomials, see for
instance~\cite[Chapter~5, Section~3, Exercise~13]{CoxLOS2015}. Note
that while the result is
only asked to be proved for the lexicographic ordering, the given hint
can be used to show that the statement holds for any monomial
ordering. Thus, any monomial in $\cT[\N^{\nu}]$ reduces to a single
monomial modulo $\cG$ and we denote by $\cT[\N^{\nu}]/I(\cC)$ the set
of monomials that cannot be reduced by $\cG$.
Furthermore, if $\by^{\bj}\in\cT[\N^{\nu}]$,
then any monomial $\by^{\bk}\in\cT[\N^{\nu}]$ that divides $\by^{\bj}$
is in $\cT[\N^{\nu}]/I(\cC)$. Indeed, if $\by^{\bk}$ were not, then it
would be a leading monomial in $I(\cC)$ and so would $\by^{\bj}$.
Hence, one can always pick a finite staircase
$T\subset\cT[\N^{\nu}]/I(\cC)$
and call the \sFGLM
algorithm with $T$ and $\prec$.
Then, by
construction, it remains to replace the output polynomials
in $\K[\by]$ by the corresponding ones in $\K[\bx]$. They
will naturally have support in $\cT[\cC]$.

\begin{example}[Continuation of Example~\ref{ex:cone}]
  It is clear that $3\ba_1=\ba_2+\ba_3$ generates all the other
  different ways to decompose an element of $\cC$, hence
  $I(\cC)=\ideal{y_1^3-y_2y_3}$. Thus, when listing the monomials
  for $\DRL(y_1\prec y_2\prec y_3)$ in $\cT[\N^{\nu}]/I(\cC)$, we will
  skip any multiple of $y_1^3$.
\end{example}

\subsection{Terms in a lattice}\label{s:guessing_lattice}
Let $\Lambda_{\pos}$ be the set of nonnegative terms of a sublattice
of $\Z^n$, we aim at guessing the recurrence relations of a table
$\bseq$ by following $\Lambda_{\pos}$. Since a lattice is a special
case of a cone, by Theorem~\ref{th:cone_sFGLM},
restricting ourselves to only considering the
subtable $\pare{\seq_{\bi}}_{\bi\in\Lambda_{\pos}}$ shall make us
guess the reduced \sgb of the ideal of relations of $\bseq$
in $\K[\Lambda_{\pos}]$.

Yet, doing so would in some way make us forget the extra structure
coming with a sublattice: namely its fundamental domain, \ie the
quotient group $\Z^n/\Lambda$. Indeed, if a set of polynomials
$\curl{f_1,\ldots,f_r}$ is such that for all $k$, there exists
$\ba_k\in\Z^n/\Lambda$ such that
$\supp f_k\subset\pare{\ba_k+\Lambda}_{\pos}$, then a classical reduced \gb
$\cG=\curl{g_1,\ldots,g_s}$ of the ideal it spans in $\K[\bx]$ satisfies the same
property. Therefore, if we expect, or even can ensure beforehand, that
the reduced \gb of the ideal of relations of $\bseq$ also satisfies
this property, we aim at guessing this \gb by working \emph{in
  parallel} on several smaller multi-Hankel matrices whose sizes have
been divided by $\Card\pare{\Z^n/\Lambda}$.

To do so, considering an input set of monomials
$T\subset\cT$, we shall split it up into
$T=\bigsqcup\limits_{\ba\in\Z^n/\Lambda}T_{\ba}$, with
$T_{\ba}=T\cap\cT[\pare{\ba+\Lambda}_{\pos}]$, and then call the \sFGLM
algorithm on $\bseq$ and $T_{\ba}$ for each $\ba$. However, the
table terms that appear in $H_{T_{\ba},T_{\ba}}$ are $\seq_{\bi}$
with $\bi\in\pare{2\ba+\Lambda}_{\pos}$. Thus, we might never
consider certain table terms. To circumvent this, we always add the row
and the column labeled with $1$ in these matrices. This yields the
\lsFGLM algorithm or Algorithm~\ref{algo:lattice_sFGLM} and
Theorem~\ref{th:lattice_sFGLM}.

\begin{algorithm2e}[htbp!]
  \small
  \DontPrintSemicolon
  \TitleOfAlgo{\lsFGLM\label{algo:lattice_sFGLM}}
  \KwIn{A table $\bseq=(\seq_{\bi})_{\bi\in\N^n}$ with coefficients in
    $\K$, a monomial ordering $\prec$, a staircase $T$ ordered
    for $\prec$, a nonnegative lattice
    $\Lambda\subseteq\Z^n$,
    a set $\cA\subseteq\N^n$ containing $0$ such that $\Lambda+\cA=\Z^n$.}
  \KwOut{A truncated reduced \gb.}
  Partition $T$ into $T=\bigsqcup\limits_{\ba\in\cA}T_{\ba}$ with
  $T_{\ba}=\pare{T\cap\cT[\pare{\ba+\Lambda}_{\pos}]}$.\;
  \Forall{$\ba\in\cA$}{
    Build the matrix $H_{\curl{1}\cup T_{\ba},\curl{1}\cup T_{\ba}}$.\;
    Compute its column profile rank $S_{\ba}$.\;
  }
  $S\coloneqq\bigcup\limits_{\ba\in\cA} S_{\ba}$.\;
  \Forall(\tcp*[f]{make $S$ a staircase}){$m\in T\setminus S$}{
    \lIf{$\exists s\in S$ such that $m\mid s\in S$}{
      $S\coloneqq S\cup\{m\}$.
    }
  }
  $L\coloneqq T\setminus S$ sorted for $\prec$.\;
  $G\coloneqq\varnothing$.\;
  \While{$L\neq\varnothing$}{
    $g\coloneqq\min_{\prec} L$\;
    Find $\ba\in\cA$ such that $g\in\cT[\pare{\ba+\Lambda}_{\pos}]$.\;
    Solve the linear system
    $H_{S_{\ba},S_{\ba}}\bgamma+H_{S_{\ba},\curl{g}}=0$.\;
    $G\coloneqq G\cup\curl{g+\sum_{s\in S_{\ba}}\gamma_s s}$.\;
    Remove $g$ and any of its multiples from $L$.
  }
  \KwRet $G$.
\end{algorithm2e}
\begin{theorem}\label{th:lattice_sFGLM}
  Let $\Lambda$ be a sublattice of $\Z^n$ with fundamental domain
  $\cA$. Let $\prec$ be a 
  monomial ordering on $\cT$ and let $T\subset\cT$ be a finite staircase
  ordered for $\prec$.

  Then, the \lsFGLM algorithm called on table $\bseq$,
  $T$ and $\prec$ returns a truncated \gb of an ideal whose polynomials are
  each with support in
  $\{1\}\cup\cT[\pare{\ba+\Lambda}_{\pos}]$, with $\ba\in\cA$.

  Furthermore, let $\cG$ be a reduced \gb for $\prec$ satisfying this support
  property. Let $\cS$ be the associated
  staircase and $\bseq$ be a generic
  \Cfin table whose ideal of 
  relations is spanned by $\cG$. Let $T$ be a staircase containing
  $\cS$ and the leading monomials of all the polynomials in $\cG$.
  Then, there exists a non empty
  Zariski open set of values for the table terms $\brac{s}_{\bseq}$ of $\bseq$, with $s\in
  \cS$, such that the \lsFGLM
  algorithm called on $\bseq$, $\prec$, $T$ and $\cA$ correctly
  guesses $\cG$.
\end{theorem}
\begin{proof}
  This proof follows mostly the same steps as that of Theorem~\ref{th:cone_sFGLM}.

  As the algorithm computes kernel vectors of matrices
  $H_{\curl{1}\cup T_{\ba},\curl{1}\cup T_{\ba}}$, the corresponding
  polynomials can only have support in
  $\curl{1}\cup\cT[\pare{\ba+\Lambda}_{\pos}]$.

  Let $S$ be the associated staircase of a reduced \gb of the ideal of
  \Crels of $\bseq$. For each $\ba\in\cA$, we let
  $S_{\ba}=S\cap\cT[\pare{\ba+\Lambda}_{\pos}]$.

  Let us show first that no monomial $m\not\in S$ is found in the
  staircase by the algorithm. As $m\in\LM(I)$, there exist
  $g=\LM(g)+\sum_{\alpha_s\in S_{\ba}}\alpha_s\,s\in I$ such that
  $\LM(g)\in\cT[\pare{\ba+\Lambda}_{\pos}]$ and $\LM(g)|m$. Thus,
  $\frac{m}{\LM(g)}g\in I$ and for all $t\in\cT$,
  $\brac{t\frac{m}{\LM(g)}g}_{\bseq}=0$. In particular, this is true for all
  $t\in T_{\bb}$, with $m\in T_{\bb}$, so that the column labeled with $m$
  is linearly dependent from the previous ones in
  $H_{\curl{1}\cup  T_{\bb},\curl{1}\cup T_{\bb}}$. Hence, neither $m$
  nor any of its multiples is in the staircase associated to the
  output. That is, the computed staircase is included in the correct
  staircase.

  It remains to prove the last statement. Let $\bseq$ be a sequence
  whose ideal of relations $I$ is spanned by $\cG$. First, from the proof of the
  \sFGLM algorithm, we know that the matrix $H_{\cS,\cS}$ has full
  rank. For any $\ba\in\cA$, we let
  $\cS_a=\{1\}\cup(\cS\cap\cT[(\ba+\Lambda)_{\pos}])$. If
  $H_{\cS_{\ba},\cS_{\ba}}$ has full rank, then for $m\in
  \LM(\cG)\cap\cT[(\ba+\Lambda)_{\pos}]$, this matrix allows us to
  determine the polynomial in $\cG$ with leading monomial $m$. Thus,
  the algorithm correctly returns $\cG$.

  For each $\ba\in\cA$, we know that the matrix $H_{\cS_{\ba},\cS}$
  has full rank. Now, to recover a relation with support in
  $\cT[\ba+\Lambda]_{\pos}$, generically, it suffices to consider
  sufficiently many shifts of this relations. And in particular, we
  can take the shifts induced by monomials in $\cS_{\ba}$, meaning
  that the matrix $H_{\cS_{\ba},\cS_{\ba}}$ generically has full rank. 
\end{proof}

\begin{remark}
  Adding a row labeled with $1$ in the matrices is necessary to prevent
  computations of incorrect relations when one of them is divisible by a non
  trivial monomial. Let us consider a unidimensional table $\bseq$
  satisfying the relation $x^4-a x^2$ with $a\in\K$ and let $\Lambda=2\Z$
  and $T=\curl{1,x,x^2,x^3,x^4}$, so that $T_0=\curl{1,x^2,x^4}$ and
  $T_1=\curl{x,x^3}$. We thus build the matrices
  \begin{align*}
  H_{T_0,T_0}&=
  \kbordermatrix{
    \ 	&1	&x^2	&x^4\\
    1	&\brac{1}_{\bseq}	&\brac{x^2}_{\bseq}	&\brac{x^4}_{\bseq}\\
    x^2	&\brac{x^2}_{\bseq}	&\brac{x^4}_{\bseq}	&\brac{x^6}_{\bseq}\\
    x^4	&\brac{x^4}_{\bseq}	&\brac{x^6}_{\bseq}	&\brac{x^8}_{\bseq}\\
  }=
  \kbordermatrix{
    \ 	&1	&x^2	&x^4\\
    1	&\seq_0	&\seq_2	&a \seq_2\\
    x^2	&\seq_2	&a \seq_2	&a^2 \seq_2\\
    x^4	&a \seq_2	&a^2 \seq_2	&a^3 \seq_2\\
    },\\
    H_{T_1,T_1}&=  \kbordermatrix{
    \ 	&x	&x^3\\
    x	&\brac{x^2}_{\bseq}	&\brac{x^4}_{\bseq}\\
    x^3	&\brac{x^4}_{\bseq}	&\brac{x^6}_{\bseq}\\
  }=
  \kbordermatrix{
    \ 	&x	&x^3\\
    x	&\seq_2	&a \seq_2\\
    x^3	&a \seq_2	&a^2 \seq_2\\
  }.
  \end{align*}

By hypothesis, clearly the column labeled with $x^4$ is linearly dependent
from the ones with label $1$ and $x^2$.
However, since $\brac{x^4-a x^2}_{\bseq}=\brac{x^6-a x^4}_{\bseq}=0$,
the column labeled with $x^3$ is linearly dependent from the
column labeled with $x$
in the second matrix. Therefore, these matrices do not allow us to
recover that $x^3$ is in the staircase of the ideal of relations of
the input table.

Yet, the matrix
\[
  H_{\{1\}\cup T_1,\{1\}\cup T_1}=
  \kbordermatrix{
    \ 	&1	&x	&x^3\\
    1	&\brac{1}_{\bseq}	&\brac{x}_{\bseq}	&\brac{x^3}_{\bseq}\\
    x	&\brac{x}_{\bseq}	&\brac{x^2}_{\bseq}	&\brac{x^4}_{\bseq}\\
    x^3	&\brac{x^3}_{\bseq}	&\brac{x^4}_{\bseq}	&\brac{x^6}_{\bseq}\\
  }=
  \kbordermatrix{
    \ 	&1	&x	&x^3\\
    1	&\seq_0	&\seq_1	&\seq_3\\
    x	&\seq_1	&\seq_2	&0\\
    x^3	&\seq_3	&0	&0\\
  }
\]
has its column labeled with $x^3$ independent from the previous two
if, and only if, $\seq_3\neq 0$, allowing us to detect that $x^3$ is
in the staircase.
\end{remark}

\begin{example}\label{ex:lattice}
  Consider the table
  $\bseq=\pare{2^i\pare{j+1\bmod 3}}_{(i,j)\in\N^2}$ defined over 
  $\Q$. Using, for instance, the \bms or the \sFGLM algorithms, we can
  easily show that its ideal of
  relations is $\ideal{y^3-1,x-2}$. Let us consider the lattice
  $\Lambda=(0,3)\Z+(1,0)\Z$, so that
  $\cA=\curl{(0,0),(0,1),(0,2)}$ and 
  $T=\curl{1,y,y^2,y^3,y^4,y^5,x}$.

  Then, Algorithm~\ref{algo:lattice_sFGLM} builds the matrices
  \[H_{T_0,T_0}={}
  \kbordermatrix{
    \ 	&1	&y^3	&x\\
    1	&1	&1	&2\\
    y^3	&1	&1	&2\\
    x	&2	&2	&4\\
  },\ H_{T_1,T_1}={}
  \kbordermatrix{
    \ 	&1	&y	&y^4\\
    1	&1	&2	&2\\
    y	&2	&0	&0\\
    y^4	&2	&0	&0\\
  },\ H_{T_2,T_2}={}
  \kbordermatrix{
    \ 	&1	&y^2	&y^5\\
    1	&1	&0	&0\\
    y^2	&0	&2	&2\\
    y^5	&0	&2	&2\\
  }.
  \]
  So that $S_0=\curl{1}$, $S_1=\curl{1,y}$ and $S_2=\curl{1,y^2}$.
  Hence $S=\curl{1,y,y^2}$,
  $L=\curl{y^3,y^4,y^5,x}$. This yields the linear systems
  $H_{S_0,S_0}\bgamma+H_{S_0,\curl{y^3}}=0$  and
  $H_{S_0,S_0}\bgamma+H_{S_0,\curl{x}}=0$ allowing us to
  recover $y^3-1$ and $x-2$.

  Notice that $\bseqq=\pare{2^i\pare{j\bmod 3}}_{(i,j)\in\N^2}$ has the same ideal of
  relations. Yet, the algorithm will build the matrices
  \[H_{T_0,T_0}={}
  \kbordermatrix{
    \ 	&1	&y^3	&x\\
    1	&0	&0	&0\\
    y^3	&0	&0	&0\\
    x	&0	&0	&0
  },\ H_{T_1,T_1}={}
  \kbordermatrix{
    \ 	&1	&y	&y^4\\
    1	&0	&1	&1\\
    y	&1	&2	&2\\
    y^4	&1	&2	&2\\
  },\ H_{T_2,T_2}={}
  \kbordermatrix{
    \ 	&1	&y^2	&y^5\\
    1	&0	&2	&2\\
    y^2	&2	&1	&1\\
    y^5	&2	&1	&1\\
  },
  \]
  so that $S_0=\varnothing$, $S_1=\curl{1,y}$, $S_2=\curl{1,y^2}$ and $S=\curl{1,y,y^2}$.
  Since the linear systems
  $H_{S_0,S_0}\bgamma+H_{S_0,\curl{x^3}}=0$ and
  $H_{S_0,S_0}\bgamma+H_{S_0,\curl{y}}=0$ are
  empty, they do not allow us to
  recover $x^3-1$ and $y-2$. Indeed,
  $\varnothing=S_0\neq S\cap\cT[\Lambda_{\pos}]=\curl{0}$.

  Yet, for the table
  $\bseqq'=\pare{\brac{(1+\lambda y)x^i y^j}_{\bseq}}$, the
  algorithm builds the matrices
  \begin{align*}
    H_{T_0,T_0}&=
                 \kbordermatrix{
                 \ 	&1		&y^3		&x\\
    1	&\lambda	&\lambda	&2\lambda\\
    y^3	&\lambda	&\lambda	&2\lambda\\
    x	&2\lambda	&2\lambda	&4\lambda
                                    },&H_{T_1,T_1}&=
                                                    \kbordermatrix{
                                                    \ 	&1		&y		&y^4\\
    1	&\lambda	&1+2\lambda	&1+2\lambda\\
    y	&1+2\lambda	&2		&2\\
    y^4	&1+2\lambda	&2		&2\\
    },\\
    H_{T_2,T_2}&=
                 \kbordermatrix{
                 \ 	&1		&y^2		&y^5\\
    1	&\lambda	&2		&2\\
    y^2	&2		&1+2\lambda	&1+2\lambda\\
    y^5	&2		&1+2\lambda	&1+2\lambda\\
    }.
  \end{align*}
  It is clear that $S_0=\curl{1}$, provided $\lambda\neq 0$,
  $S_1=\curl{1,y}$, provided $4\lambda^2+2\lambda+1\neq 0$, and
  $S_2=\curl{1,y^2}$, provided $2\lambda^2+\lambda-4\neq 0$. All in all,
  the algorithm succeeds for $\bseqq'$ as long as $\lambda$ does not satisfy
  $\lambda(4\lambda^2+2\lambda+1)(2\lambda^2+\lambda-4)=0$.
\end{example}

\begin{remark}
  While we assume that $\Lambda$ is a sublattice of $\Z^n$, hence of
  rank $n$, it can actually be any $\Z$-submodule of smaller
  rank $\nu$. However,
  this means we can only guess an ideal of relations in $\nu$
  variables so that it may not be the whole ideal of relations.
  Nevertheless, this kind of restriction can be of
  interest in the \Pfin application where the kernel equation makes us
  study the \Pfin nature of a subsequence where some indices are set.
\end{remark}

\subsection{Application to \gb change of orderings
  with the action of a matrix group}\label{ss:app_fglm}
In~\cite{FaugereM2017}, the authors propose a variant of the \FGLM
algorithm~\cite{FaugereGLM1993}, the so-called \spFGLM algorithm,
relying on guessing \Crels. More
precisely, from the input \gb $\cG$, they build a random table $\bseq$
whose
ideal of relations is $\ideal{\cG}$. To do so, first, for each monomial $s$ in the staircase associated to
$\cG$, they pick at random the table term
$\brac{s}_{\bseq}$, then they compute the other table terms using the
\Crels induced by $\cG$. Finally, applying an algorithm for guessing \Crels on this
table and the second input ordering, they obtain the \gb of the ideal
of relations of this table for this second ordering. If the first \gb
spans a Gorenstein ideal~\cite{Brachat2010,ElkadiM2007}, then with
high probability, the output \gb is a \gb of the same ideal and thus
the target one.

In particular, assuming generic properties, detailed below, on the polynomials that
span the ideal we want to compute a \gb of, this algorithm comes down
to computing products of a sparse matrix and some vectors and solving
Hankel systems.

The goal of this section is to extend this approach to abelian group actions
on the ideal. In particular, we will restrict ourselves to finite
abelian matrix
group actions, that is finite abelian subgroups of $\GL(n)$ where $A\in\GL(n)$ acts
on $f(\bx)\in\K[\bx]$ by sending it to
$f(A\bx)$.

\subsubsection{Finite matrix group actions}\label{sss:gln}
We start by recalling some results on finite matrix group actions on
ideals of $\K[\bx]$.

Let $G$ be a finite abelian matrix group. By the invariant factors theorem, there
exist $q_1\mid\cdots\mid q_{\ell}$ such that $G\simeq
\Z/q_1\Z\times\cdots\times\Z/q_{\ell}\Z$ and in particular, for any $g\in
G$, $g^{q_{\ell}}=1$ and $q_{\ell}$ is minimal for this property.

Furthermore if $|G|=q_1\cdots q_{\ell}$ is
not divisible by the characteristic of the coefficient field $\K$,
then there exists a primitive $q_{\ell}$th root of unity $\zeta$ such
that the matrices in $G$ are simultaneously diagonalizable with powers of 
$\zeta$ on the diagonals,
see~\cite[Theorem~2]{FaugereS2013}. After this diagonalization
process, which comes down to a linear change of variables, for each
matrix in $G$, there exist
natural numbers $0\leq\veps_1,\ldots,\veps_n\leq q_{\ell}-1$
such that $x_i$ is sent onto $\zeta^{\veps_i}x_i$ by this matrix.

\begin{definition}[{\cite[Definition~3]{FaugereS2013}}]\label{def:G-deg}
  Let $G\simeq
  \Z/q_1\Z\times\cdots\times\Z/q_{\ell}\Z$, with $q_1\mid\cdots\mid
  q_{\ell}$, be a diagonal subgroup of $\GL(n)$ and $\zeta$ be a
  $q_{\ell}$th root of unity, then there exist matrices
  $D_1,\ldots,D_n$ spanning $G$ such that each $D_i$ has order $q_i$.

  For each monomial $m\in\cT$, there exist
  $(\mu_1,\ldots,\mu_{\ell})\in\Z/q_1\Z\times\cdots\times\Z/q_{\ell}\Z$
  such that for all $i$, $m$ is sent onto 
  $\zeta^{\mu_i q_{\ell}/q_i}m$ by $D_i$. Then, $m$ is said to have
  \emph{$G$-degree} $\pare{\mu_1,\ldots,\mu_{\ell}}$.

  Furthermore, a polynomial is \emph{$G$-homogeneous} if all its monomials
  have same $G$-degree.
\end{definition}

From this, one can prove that the $G$-degree of the product of two
monomials is the sum of their $G$-degrees. Since the $G$-degree of the
monomial $1$ is $\pare{0,\ldots,0}$, the subset of monomials of
$G$-degree $\pare{0,\ldots,0}$ is a sublattice $\cT[\Lambda_{\pos}]$ of $\cT$. A
consequence of this is that if $f_1,\ldots,f_s$ are $G$-homogeneous
polynomials, then a reduced \gb of $\ideal{f_1,\ldots,f_s}$ is made of
$G$-homogeneous polynomials as well and $\ideal{f_1,\ldots,f_s}$ is
stable by the action of $G$, see~\cite[Theorem~4]{FaugereS2013}.

\subsubsection{\gbs change of orderings}
From the reduced $\DRL$ \gb $\cG$ of such an ideal, it then makes sense to apply the
\spFGLM algorithm in order to obtain the reduced $\LEX$ \gb.
Since we already know that the
support of each polynomial in the target \gb lies on a lattice, or an
affine translate thereof, we can use
Algorithm~\ref{algo:lattice_sFGLM} to guess the relations on the table
that is built by the algorithm. Furthermore,
since the table is built with its first table terms picked at random,
no fake relations, like in Example~\ref{ex:lattice}, should be
guessed.

\begin{proposition}
  Let $G$ be an abelian group as in
  Definition~\ref{def:G-deg}.
  
  Let $f_1,\ldots,f_s$ be generic polynomials of degree
  $d_1,\ldots,d_s$ such that $I=\langle f_1,\ldots,f_s\rangle$ is
  zero-dimensional stable by the action of $G$.
  
  Let $\prec$ and $<$ be a monomial orders and let $\cG$ and $\cH$ be
  the reduced \gb of $I$ for $\prec$ and $<$respectively.

  Then, the guessing step of $\cH$ when calling the \spFGLM algorithm
  on $\cG$, $\prec$ and $<$ can be sped up using the \lsFGLM
  algorithm by a factor $O(|G|^{\omega-1})$ instead of the \sFGLM
  algorithm.
\end{proposition}
\begin{proof}
  From~\cite{FaugereS2013} and the genericity assumption on $I$, the
  polynomials in $\cH$ are evenly split between all the
  $G$-degrees. Furthermore, so are the monomials in the
  $\Staircase(\cH)$.

  Now, to recover $\cH$, one needs to call the \sFGLM and \lsFGLM
  algorithms on a staircase $T$ that contains $S$ and
  $\LM(\cH)$. Then, the \sFGLM algorithm computes the right-kernel of
  $H_{T,T}$ in at most $O(\#T^{\omega})$ operations. Now, the \lsFGLM
  will build $|G|$ submatrices of $H_{T,T}$ of size roughly $\#T/|G|$
  and compute the right-kernel of each. Thus it can be done in
  $O(\#T^{\omega}/|G|^{\omega-1})$.
\end{proof}

We shall say that a zero-dimensional ideal $I\subset\K[\bx]$
has
\begin{description}
\item[Property~S,] if its reduced \gb for
  $\LEX(x_n\prec\cdots\prec x_1)$ is in
  \emph{shape position}. That is, there exist
  $g_1,\ldots,g_n\in\K[x_n]$ of degree at most $D-1$ such that this
  reduced \gb is
  $\curl{x_n^D+g_n(x_n),x_{n-1}+g_{n-1}(x_n),\ldots,x_1+g_1(x_n)}$.
\item[Property~M,] if its reduced \gb for
  $\DRL(x_n\prec\cdots\prec x_1)$
  satisfies the following condition. For every monomial $m$ in the
  staircase associated to this \gb, either $m x_n$ is in the staircase
  or it is the leading monomial of some polynomial in this \gb.
\end{description}
Let us recall that if $I$ is spanned by generic polynomials
$f_1,\ldots,f_n\in\K[\bx]$ of degree $d_1,\ldots,d_n$ and $\K$ is
sufficiently large or infinite, then both Properties~S and~M are satisfied. See for
instance~\cite[Proposition~5.3]{FaugereM2017}, where $x_1$ is chosen
as the smallest variable, for the latter. For the former, this is a
direct consequence of $I$ being radical with solutions not sharing the
same last coordinate. Thus, the Shape lemma
applies without requiring any change of variables, see~\cite[Lemma~1.4]{GianniM1989}.

Under these assumptions, several algorithms can be used to compute the
reduced $\LEX$ \gb of an ideal of degree $D$ from the reduced $\DRL$
one. The seminal one, 
\FGLM~\cite{FaugereGLM1993} with a complexity $O(n D^3)$, the \spFGLM
one~\cite{FaugereM2011,FaugereM2017} with a complexity
$\tilde{O}(k D^2+n D)$, where $k$ is the number of polynomials in the
reduced $\DRL$ \gb whose leading monomial is divisible by $x_n$, a
faster variant~\cite{FaugereGHR2014} of the \FGLM algorithm using Keller-Gehrig
algorithm~\cite{Keller-Gehrig1985} 
or
\textsc{SyzygyModuleBasis}~\cite[Algorithm~3]{NeigerS2020} both with a complexity
$\tilde{O}(n D^{\omega})$.

Whenever an ideal is stabilized by the action of such a finite abelian
matrix group, the goal is to take advantage of this to speed the
change of orderings algorithm up. In~\cite[Theorem~10]{FaugereS2013}, the
authors show the complexity of the \FGLM algorithm drops to
$O(D^3/|G|^2)$, mainly because they deal with $|G|$ matrices of sizes roughly
$D/|G|$ instead of one larger matrix of size $D$. These matrices
correspond to those of monomials of each $G$-degree. It would be
interesting to study if, using the same
trick, one could make the complexities of the faster variant of the
\FGLM algorithm or of the \textsc{SyzygyModuleBasis} algorithm
drop to $\tilde{O}(n D^{\omega}/|G|^2)$ or even $\tilde{O}(n
D^{\omega}/|G|^{\omega-1})$.

Let us notice that in this situation, the \spFGLM algorithm only
relies on $1$-dimensional algorithms like the \bmy one as the best
strategy.
We now focus
on the complexity improvements one can reach in this setting when the
ideal spanned by $\cG$ and $\cH$ is stable under the action of a $G$
as in Definition~\ref{def:G-deg}.

We now focus on the \spFGLM algorithm and we assume that a reduced
$G$-homogeneous \gb for
$\DRL(x_n\prec\cdots\prec x_1)$, spanning an ideal satisfying
Property~M, is given and the goal is to recover the reduced \gb for
$\LEX(x_n\prec\cdots\prec x_1)$ satisfying Property~S. By $G$-homogeneity, the
support of each polynomial in the target \gb, namely
$\{x_n^D+g_n(x_n),x_{n-1}+g_{n-1}(x_n),\ldots,x_1+g_1(x_n)\}$, is already
known. It is given by the $G$-degree of its
leading monomial, namely 
$x_n^D,x_{n-1},\ldots,x_1$. Since $G$ is finite, there exists $d>0$ minimal such
that $x_n^d$ has $G$-degree $\pare{0,\ldots,0}$ and there exists
$\delta_n,\ldots,\delta_1\geq 0$, all minimal, such that $x_n^{\delta_n}$ has same
$G$-degree as $x_n^D$ and $x_n^{\delta_i}$ has same $G$-degree as
$x_i$ for $1\leq i\leq n-1$. Therefore,
for $1\leq i\leq n$,
$\supp g_i=\curl{x_n^{\delta_i},x_n^{\delta_i+d},\ldots,
  x_n^{\delta_i+\floor{\frac{D-1-\delta_i}{d}}d}}$.

Thus, the polynomial $g_n$ can be computed by solving the following
Hankel system
\[
  \kbordermatrix{
    \	 		&x_n^{\delta_n}	&x_n^{\delta_n+d}	&\cdots	&x_n^{D-d}\\
    x_n^{\delta_n}	&\brac{x_n^{2\delta_n}}_{\bseq}
    &\brac{x_n^{2\delta_n+d}}_{\bseq}	&\cdots	&\brac{x_n^{D-d+\delta_n}}_{\bseq}\\
    x_n^{\delta_n+d}	&\brac{x_n^{2\delta_n+d}}_{\bseq}	&\brac{x_n^{2\delta_n+2d}}_{\bseq}
    &\cdots	&\brac{x_n^{D-d+\delta_n}}_{\bseq}\\
    \vdots	&\vdots	&\vdots	&&\vdots\\
    x_n^{\delta_n+\floor{\frac{D-1-\delta_n}{d}}d}=x_n^{D-d}
    &\brac{x_n^{D-d}}_{\bseq}
    &\brac{x_n^{D-d+\delta_n}}_{\bseq}
    &\cdots	&\brac{x_n^{2D-2d}}_{\bseq}\\
  }\bgamma+
  \kbordermatrix{
    \ 	&x_n^D\\
    x_n^{\delta_n}	&\brac{x_n^{D+\delta_n}}_{\bseq}\\
    x_n^{\delta_n+d}	&\brac{x_n^{D+\delta_n+d}}_{\bseq}\\
    \vdots	&\vdots\\
    x_n^{D-d}		&\brac{x_n^{2D-d}}_{\bseq}\\
  }=0.
\]
Denoting $M_n$ the matrix of the multiplication by $x_n$ in
$\K[\bx]/I$, $\bone=\pare{
  \begin{smallmatrix}
    1\\0\\\vdots\\0
  \end{smallmatrix}}$ and $\br$ a vector picked at random, the
table terms $\brac{x_n^i}_{\bseq}$ are defined as
$\br^{\T}M_n^i\bone$. This is done by computing
$v_0=\br^{\T}, v_1=v_0M_n,v_2=v_1M_n,\ldots$ and then extracting the
first coordinate of each vector to simulate the multiplication by
$\bone$.

Since we do not need all the terms but only $v_{2\delta_n},
v_{2\delta_n+d}, v_{2\delta_n+2d},\ldots$, we first compute
$v_{2\delta_n}$ and $M_n^d$ in order to perform big steps. Let us
notice that, following~\cite{FaugereM2011,FaugereM2017}, by Property~M,
the columns of matrix $M_n$ are of two types. If a monomial $m$ in the
staircase is
such that $m x_n$ is still in the staircase, then the column
corresponding to $m$ is \emph{trivial}, it is a vector of the
canonical basis. Otherwise, $m$ is the leading monomial of a
polynomial $g$ in the reduced \gb and the column corresponding to $m$
is the coefficient vector of its normal form, namely $m-g$. Usually,
these latter vectors are denser than the former.
Then,
$M_n^d$ has the same shape as $M_n$, it has trivial and non-trivial
columns. Furthermore, if $M_n$ has $k$ non-trivial columns, then
$M_n^d$ has at most $\max(D,k d)$ non-trivial columns.
From~\cite{FaugereS2013} and the genericity
assumption on $I$, we know we can split $M_n$ in
$|G|^2$ matrices of size at most $\ceil{D/|G|}$. Furthermore, its
non-trivial columns are evenly split in
the small matrices, \ie the number of non-trivial
columns of each small matrix is at most $\ceil{k/|G|}$. Then, we can
multiply all these small matrices accordingly to obtain the splitting
of $M_n^d$.

Now, polynomials $g_1,\ldots,g_{n-1}$ can be computed by solving a
similar Hankel system:
\[
  \kbordermatrix{
    \	 		&x_n^{\delta_i}	&x_n^{\delta_i+d}
    &\cdots	&x_n^{\delta_i+\floor{\frac{D-1-\delta_i}{d}}d}\\
    x_n^{\delta_n}	&\brac{x_n^{\delta_i+\delta_n}}_{\bseq}
    &\brac{x_n^{\delta_i+\delta_n+d}}_{\bseq}
    &\cdots	&\brac{x_n^{\delta_i+\delta_n+\floor{\frac{D-1-\delta_i}{d}}d}}_{\bseq}\\
    x_n^{\delta_n+d}	&\brac{x_n^{\delta_i+\delta_n+d}}_{\bseq}
    &\brac{x_n^{\delta_i+\delta_n+2d}}_{\bseq}
    &\cdots	&\brac{x_n^{\delta_i+\delta_n+\floor{\frac{D-1-\delta_i}{d}+1}d}}_{\bseq}\\
    \vdots	&\vdots	&\vdots	&&\vdots\\
    x_n^{D-d}	&\brac{x_n^{\delta_i+D-d}}_{\bseq}	&\brac{x_n^{\delta_i+\delta_n+D-d}}_{\bseq}
    &\cdots	&\brac{x_n^{\delta_i+\delta_n+\floor{\frac{D-1-\delta_i}{d}-1}d+D}}_{\bseq}\\
  }\bgamma+
  \kbordermatrix{
    \ 	&x_i\\
    x_n^{\delta_n}	&\brac{x_ix_n^{\delta_n}}_{\bseq}\\
    x_n^{\delta_n+d}	&\brac{x_ix_n^{\delta_n+d}}_{\bseq}\\
    \vdots	&\vdots\\
    x_n^{D-d}		&\brac{x_ix_n^{D-d}}_{\bseq}\\
  }=0.
\]
However, the matrices might all be different. In order to speed the
computation up, we change the linear systems into ones with the same
matrix as the first one. This is done by multiplying all the column
labels by $x_n^{\delta_n-\delta_i}$. The constant vectors of the
systems thus become
\[
  \kbordermatrix{
    \ 	&x_n^{\delta_n-\delta_i}x_i\\
    x_n^{\delta_n}	&\brac{x_ix_n^{2\delta_n-\delta_i}}_{\bseq}\\
    x_n^{\delta_n+d}	&\brac{x_ix_n^{2\delta_n+d-\delta_i}}_{\bseq}\\
    \vdots	&\vdots\\
    x_n^{D-d}		&\brac{x_ix_n^{D-d+\delta_n-\delta_i}}_{\bseq}\\
  }.
\]

\begin{proposition}\label{prop:FGLM}
  Let $I\subset\K[\bx]$ be a zero-dimensional ideal of degree $D$,
  invariant under the action of a finite diagonal matrix
  group $G$. Let us assume that $I$ satisfies both properties~S and~M
  and that the matrix $M_n$ has $k$ non-trivial columns. Let furthermore
  $S$ be the staircase associated to the $\LEX(x_n\prec\cdots\prec x_1)$
  \gb of $I$,
  $\cT[\Lambda_{\pos}]$ be the set of monomials of $G$-degree $0$ and for
  $A,B\subseteq\cT$, $A+B=\curl{ab\middle|a\in A,b\in B}$
  be the Minkowski sum of $A$ and $B$.

  Then, we can recover the $\LEX(x_n\prec\cdots\prec x_1)$ \gb, $\cG$,
  of $I$ from
  its $\DRL(x_n\prec\cdots\prec x_1)$ \gb using
  $\Card{\pare{(S\cap\cT[\Lambda_{\pos}])+
      \big((S\cap\cT[\Lambda_{\pos}])\cup\LM(\cG)\big)}}$
  table 
  terms and $O\pare{\frac{n k D^2}{|G|}}$ operations.
\end{proposition}
\begin{proof}
  Since the ideal $I$ satisfies Property~S, the staircase $S$ associated to its
  $\LEX(x_n\prec\cdots\prec x_1)$ \gb is
  $\curl{1,x_n,\ldots,x_n^{D-1}}$. Therefore, by definition of $d$,
  $S\cap\cT[\Lambda_{\pos}]=\curl{1,x_n^d,\ldots,x_n^{\floor{\frac{D-1}{d}}d}}$. Thus,
  the matrix rows labels are in bijection with a subset of
  $S\cap\cT[\Lambda_{\pos}]$ while the matrix column labels and the
  column-vector label are in bijection with a subset of
  $\pare{S\cap\cT[\Lambda_{\pos}]}\cup\LM(\cG)$. This show that only
  $\Card{\pare{(S\cap\cT[\Lambda_{\pos}])+
      \big((S\cap\cT[\Lambda_{\pos}])\cup\LM(\cG)\big)}}$
  table terms are required.

  Since $I$ also satisfies Property~M, $M_n$ has $k$ non-trivial
  columns and $D-k$ columns that are vectors of the canonical
  basis. Furthermore, these non-trivial columns correspond
  to $G$-homogeneous polynomials, so each of them has at most
  $O(D/|G|)$ nonzero coefficients. Thus, $M_n$ has $O(k D/|G|)$ nonzero
  coefficients.
  Now, computing $v_{2\delta_n}$ requires $2\delta_n$ multiplications
  between $M_n$ and a vector. Hence
  $v_{2\delta_n}$ can be
  computed in $O\pare{\delta_n k D/|G|}$ operations.

  It remains to compute
  $v_{2\delta_n+j d}=\br^{\T}M_n^{2\delta+j d}$ for all $j$ up to
  $(2D-d-2\delta_n)/d$ by successive
  multiplications by $M_n^d$. While $M_n^d$ has $\max(D,k d)$ non-trivial
  columns, these non-trivial columns still represent $G$-homogeneous
  polynomials, thus $M_n^d$ has $O(k d D/|G|)$ nonzero coefficients. Hence,
  all these vectors can be computed in
  $O(k D^2/|G|)$ operations.

  For the constant vectors of the Hankel systems, we need to
  extract the coefficients 
  corresponding to $x_i$ of vectors $v_{2\delta_n-\delta_i+j d}$ for
  each $i$ and each $j$. First, let us notice that each vector
  $v_{2\delta_n-\delta_i}$ has been computed in order to obtain
  $v_{2\delta_n}$. Then, the others are computed by successive
  multiplications by $M_n^d$, as for the vectors
  $v_{2\delta_n+j d}$. Thus, they can be obtained in $O(n k D^2/|G|)$ operations.

  Finally, these linear systems are Hankel of size $O(D/d)$ sharing
  the same matrix and thus can be solved
  in $O\pare{\sM\pare{\frac{D}{d}}\pare{n+\log\frac{D}{d}}}$ operations,
  see~\cite{BrentGY1980}. This step is therefore not the bottleneck of the algorithm.
\end{proof}


\section{Adaptive approach}\label{s:adaptive}
\subsection{The \asFGLM algorithm}
A drawback of the \sFGLM algorithm is that, in order to return the
correct \gb, it needs to be called with a staircase $T$ that contains
both the support of the \gb and its associated staircase. Without the
help of an oracle, which we have in a multi-modular setting for
instance, it is not an easy task to find such a $T$. Thus, an adaptive
variant was designed by the authors
in~\cite{BerthomieuBF2015,BerthomieuBF2017} in order to discover the
associated staircase and the \gb step by step. As a byproduct, it also
minimizes the number of table queries.

Given a table $\bseq$ and a monomial ordering $\prec$, the
\asFGLM starts
with the empty set
$S=\emptyset$. At each step, $S$ is a staircase and a subset of
  the correct one. Then, for a monomial $\bx^{\bi}$ such that
  $S\cup\curl{\bx^{\bi}}$ is also a staircase, if
$H_{S\cup\curl{\bx^{\bi}},S\cup\curl{\bx^{\bi}}}$
has a greater rank than $H_{S,S}$, then $S$ is replaced by
$S\cup\curl{\bx^{\bi}}$. Otherwise we have found a relation with leading
monomial $\bx^{\bi}$ and we shall never try any multiple of
$\bx^{\bi}$ as a new term in $S$.

In the cone setting, as in Section~\ref{ss:guessing_cone}, the two
strategies can be used. If we build an auxiliary table
$\bseqq\in\K^{\N^{\nu}}$, then the \asFGLM algorithm can directly be
called on $\bseqq$ provided we only try to add monomials $\by^{\bj}$
that are in $\cT[\N^{\nu}]/I(\cC)$. If we rather call it on the
original table $\bseq\in\K^{\N^n}$, then we modify the algorithm so
that only monomials in $\cT[\cC]$ are used. Furthermore, once a
relation with leading monomial $\bx^{\bi}$ is found, we shall never
try any multiple $\bx^{\bi+\bj}$ in the cone, \ie with
$\bx^{\bj}\in\cT[\cC]$.
\begin{example}
  Consider the linear King walk
  $\bseq=\pare{\seq_{i_0,i_1}}_{(i_0,i_1)\in\N^2}$ counting the number
  of ways to reach $i_1$ in $i_0$ steps of size $1$ starting from $0$ in the
  nonnegative ray. It is clear that $\seq_{i_0,i_1}=0$ whenever
  either $i_1>i_0$ or $i_0+i_1=1\bmod 2$, so that we shall only consider the cone
  \begin{align*}
    \cC &= \curl{(i_0,i_1)\in\N^2\middle|
          i_0+i_1=0\bmod 2,i_1\leq i_0}\\
        &= (1,1)\N+(0,2)\N.
  \end{align*}
  Assume we consider
  the $\LEX(x_1\prec x_0)$ ordering, so that
  $\cT[\cC]=\curl{1,x_0x_1,x_0^2,x_0^2x_1^2,x_0^4,\ldots}$.
  \begin{enumerate}
  \item We build the matrix $\kbordermatrix{
      \	&1\\
      1	&1\\
    }$ which has full rank.
  \item We increase the matrix by adding
    monomials in $\cT[\cC]$ so we build $\kbordermatrix{
      \		&1	&x_0x_1\\
      1		&1	&1\\
      x_0x_1	&1	&1\\
    }$ which does not have full rank, so we have found the fake relation
    $x_0x_1-1$.
  \item We increase the matrix to
    $\kbordermatrix{
      \		&1	&x_0^2\\
      1		&1	&1\\
      x_0^2	&1	&2\\
    }$ which has full rank.
  \item We increase the matrix to
    $\kbordermatrix{
      \		&1	&x_0^2	&x_0^4\\
      1		&1	&1	&2\\
      x_0^2	&1	&2	&5\\
      x_0^4	&2	&5	&14\\
    }$ which has full rank.
  \item And so on.
  \end{enumerate}
\end{example}

In the lattice setting however, we need to be more careful. We shall
make one matrix per element in $\Z^n/\Lambda$ and 
each time we must add an extra column and an extra row, they will be
added to the matrix corresponding to the monomial labeling the extra
column. If there is no rank increase, then as usual a relation is
found and no multiple of this monomial will ever label any new column in
\emph{any} matrix.
This yields Algorithm~\ref{algo:lattice_asFGLM} and Theorem~\ref{th:lasFGLM}.
\begin{algorithm2e}[htbp!]
  \small
  \DontPrintSemicolon
  \TitleOfAlgo{\lasFGLM\label{algo:lattice_asFGLM}}
  \KwIn{A table $\bseq=(\seq_{\bi})_{\bi\in\N^n}$ with coefficients in
    $\K$, a monomial ordering $\prec$, a nonnegative lattice
    $\Lambda\subseteq\N^n$,
    a set $\cA\subseteq\N^n$ containing $0$ such that $\Lambda+\cA=\Z^n$.}
  \KwOut{A set $G$ of relations.}
  \lIf{$\seq_{(0,\ldots,0)}=0$}{\KwRet $[1]$.}
  $L \coloneqq \curl{x_1,\ldots,x_n}$.\;
  Sort $L$ by increasing order \wrt $\prec$.\;
  $G \coloneqq \emptyset$ \tcp*{the future set of relations}
  \lForall(\tcp*[f]{the future staircase}){$\ba\in\cA$}{
    $S_{\ba}\coloneqq \curl{1}$.
  }
  \While{$L\neq\emptyset$}{
    $m\coloneqq $ first element of $L$ and remove it from $L$.\;
    Pick $\ba\in\cA$ such that $m\in\cT[\pare{\ba+\Lambda}_{\pos}]$.\;
    $S' \coloneqq S_{\ba}\cup\curl{m}$.\;
    \uIf(\tcp*[f]{No relation}){$H_{S',S'}$ has full rank}{
      $S_{\ba} \coloneqq S'$.\;
      $L \coloneqq L\cup\curl{x_1m,\ldots,x_n m}$
      Sort $L$ by increasing order \wrt $\prec$ and remove duplicates
      and multiples of $\LM(G)$.
    }
    \Else(\tcp*[f]{Relation!}){
      Solve $H_{S_{\ba},S_{\ba}}\bgamma+H_{S_{\ba},\curl{m}}=0$.\;
      $G \coloneqq G\cup\curl{m+\sum_{s\in S_{\ba}}\gamma_s s}$ and
      remove multiples of $m$ in $L$.\;
    }
  }
  \Return $G$.
\end{algorithm2e}

\begin{theorem}\label{th:lasFGLM}
  Let $\Lambda$ be a sublattice of $\Z^n$ with fundamental domain
  $\cA$. Let $\prec$ be a 
  monomial ordering on $\cT$.
  Let us assume that the \lasFGLM algorithm called on table $\bseq$,
  $\prec$, $\Lambda$ and $\cA\subseteq\N^n$ returns a non-empty set of
  polynomials $G$.

  Let us denote by $S$ the associated staircase to
  $G$ and $S_{\ba}=S\cap\cT[\pare{\ba+\Lambda}_{\pos}]$ for each
  $\ba\in\cA$.

  Then, for any polynomial $g\in G$ with
  $\LM(g)\in\cT[\pare{\ba+\Lambda}_{\pos}]$ and any $s\in S_{\ba}$
  with $s\prec\LM(g)$, we have $\brac{gs}_{\bseq}=0$.

  Furthermore, let $\cG$ be a \gb for $\prec$ spanning a
  $0$-dimensional ideal such that for all $g\in
  \cG$, there exists $\ba\in\cA$ such that $\supp
  g\subset\cT[\pare{\ba+\Lambda}_{\pos}]$. Let $\cS$ be the associated
  staircase and $\bseq$ be a generic
  \Cfin table whose ideal of 
  relations is spanned by $\cG$. Then, there exists a non empty
  Zariski open set of values for the table terms $\brac{s}_{\bseq}$ of $\bseq$, with $s\in
  \cS$, such that
  the \lasFGLM
  algorithm called on $\bseq$, $\prec$ and $\cA$ correctly
  guesses $\cG$.
\end{theorem}
\begin{proof}
  In the while loop, either monomial $m$ is added to the
  staircase $S_{\ba}$ or it is the leading monomial of a polynomial
  $g$ that is added to $G$.

  In the latter case, only monomials less than $m$ can have been added
  to $S_{\ba}$. Thus, the current set $S_{\ba}$ is actually the
  final set $S_{\ba}$ with only elements less than $m$, \ie
  $S_{\ba}\cap\curl{t\prec m}$. Now,
  $H_{S_{\ba}\cap\curl{t\prec m},S_{\ba}\cap\curl{t\prec m}}\gamma
  +H_{S_{\ba}\cap\curl{t\prec m},\curl{m}}=0$ is equivalent to $\brac{gs}_{\bseq}=0$
  for any $s$ a row index, that is $s\in S_{\ba}$ with $s\prec m$.
  
  Let us prove the second assertion. For any $\ba\in\cA$, let
  $\cS_{\ba}=\cS\cap\cT[\pare{\ba+\Lambda}_{\pos}]$. A necessary and sufficient
  condition for the \lasFGLM algorithm to correctly guess $\cG$ is
  that for each $\ba$, $\cS_{\ba}\subseteq S_{\ba}$, which means that
  $S_{\ba}=\curl{1}\cup\cS_{\ba}$. This can only happen if, for each
  $\ba$ and each monomial $m\in\cS_{\ba}$, the rank condition in the
  if statement is
  fulfilled.
  Following the proof of Theorem~\ref{th:lattice_sFGLM}, we can build
  a sequence $\bseqq$ from $\bseq$ whose ideal of relations is also
  spanned by $\cG$ but whose such that the rank conditions in the if
  statement is satisfied for all monomial $m\in\cS$.
\end{proof}

\begin{remark}
  If an incorrect staircase is guessed, then not much can be said on
  the output set of polynomials compared to the correct \gb. However,
  we know that the guessed staircase is included in the correct one.
\end{remark}

\begin{example}
  Let us consider the same first table as in Example~\ref{ex:lattice},
  $\bseq=\pare{2^i\pare{j+1\bmod 3}}_{(i,j)\in\N^2}$ and its
  associated lattice $\Lambda=(0,3)\Z+(1,0)\Z$, so that
  $\cA=\curl{(0,0),(0,1),(0,2)}$. We also consider the $\LEX(y\prec
  x)$ ordering.
  \begin{enumerate}
  \item We build three matrices $\kbordermatrix{
      \	&1\\
      1	&1\\
    }$ which have full rank.
  \item We increase the second matrix to $\kbordermatrix{
      \		&1	&y\\
      1		&1	&2\\
      y		&2	&0\\
    }$ which has full rank.
  \item We increase the third matrix to
    $\kbordermatrix{
      \		&1	&y^2\\
      1		&1	&0\\
      y^2	&0	&2\\
    }$ which has full rank.
  \item We increase the first matrix to
    $\kbordermatrix{
      \		&1	&y^3\\
      1		&1	&1\\
      y^3	&1	&1\\
    }$ which does not have full rank so that we have found that $y^3-1$ is in the
    ideal of relations.
  \item We increase the first matrix to
    $\kbordermatrix{
      \		&1	&x\\
      1		&1	&2\\
      x		&2	&4\\
    }$ which does not have full rank so that we have found that $x-2$ is in the
    ideal of relations.
  \item We return $\curl{y^3-1,x-2}$.
  \end{enumerate}

\end{example}

\subsection{Mixed approach for guessing \Prels}
In~\cite{BerthomieuF2016}, the authors proposed a mixed approach for
guessing \Prels based on a \gb computations for reducing the number of
table queries. The idea is that if two polynomials
$g_1,g_2\in\K[\bt]\ideal{\bx}$ are \Prels satisfied by the table, then
any polynomial in $\ideal{g_1,g_2}$ is also a \Prel. Therefore, as
soon as two \Prels $g_1$ and $g_2$ are guessed, the goal is to
compute a \gb $\curl{g_1,g_2,\ldots,g_r}$ of $\ideal{g_1,g_2}$. This
will yield polynomials, namely $g_3,\ldots,g_r$, whose leading monomials are not in
$\ideal{\LM(g_1),\LM(g_2)}$. The advantage of this method is
twofold. First, since
$\LM(g_3),\ldots,\LM(g_r)\succ\LM(g_1),\LM(g_2)$, they require more
queries to the table to be correctly guessed. Yet, such a \gb
computation does not require any more queries. Then, these \Prels may
help us determine that the ideal of \Prels is $0$-dimensional in
$\K(\bt)\ideal{\bx}$. This is a necessary condition for the table to
be \Pfin.

The aim of this section is to extend this approach for guessing \Prels
of a table when only considering terms in a cone or when the ideal of
relations is stable by the action of a subgroup of $\GL(n)$.

\begin{lemma}\label{lem:sgb_skewpol}
  Let $\cT[\cC]$ be a cone of monomials in $x_1,\ldots,x_n$, as before.
  Let us assume that $f_1,f_2\in\K[\bt]\ideal{\bx}$ are both polynomials
  with monomials in
  $\cT[\N^n]\times\cT[\cC]=\curl{\bt^{\bk}\bx^{\bi}\middle|\bx^{\bi}\in\cT[\cC]}$.
  Then, any polynomial
  $f_1a_1+f_2a_2$ in the right ideal $\ideal{f_1,f_2}$, such that
  $\supp a_1,\supp a_2\in\cT[\N^n]\times\cT[\cC]$, has its
  support in $\cT[\N^n]\times\cT[\cC]$ as well.

  In particular, we can compute a
  \sgb of $\ideal{f_1,f_2}$ 
  with monomials all in $\cT[\N^n]\times\cT[\cC]$ using Buchberger's
  algorithm or Faug\`ere's \Fquatre algorithm, restricted to only
  multiplying the polynomials by monomials in $\cT[\N^n]\times\cT[\cC]$.
\end{lemma}
\begin{proof}
  We need to prove that if $\supp f$ and $\supp a$ are in
  $\cT[\N^n]\times\cT[\cC]$, then so is $\supp fa$. By linearity, this
  comes down to proving that if two monomials $\bt^{\bl}\bx^{\bj}$ and
  $\bt^{\bk}\bx^{\bi}$ are in $\cT[\N^n]\times\cT[\cC]$, then so is
  the support of their product. Since
  \begin{align*}
    \bt^{\bl}\bx^{\bj}\bt^{\bk}\bx^{\bi}
    &= \bt^{\bl}\pare{\bt-\bj}^{\bk}\bx^{\bj+\bi}\\
    &= \sum_{q_1,\ldots,q_n=0}^{\ell_1,\ldots,\ell_n}
      \binom{k_1}{q_1}\cdots\binom{k_n}{q_n}(-j_1)^{k_1-q_1}\cdots(-j_n)^{k_n-q_n}
      t_1^{\ell_1+q_1}\cdots t_n^{\ell_n+q_n}\bx^{\bj+\bi}
  \end{align*}
  and $\bx^{\bj+\bi}\in\cT[\cC]$,
  $\bt^{\bl}\bx^{\bi}\bt^{\bk}\bx^{\bi}\in\cT[\N^n]\times\cT[\cC]$.

  Now, in
  $\cT[\N^n]\times\cT[\cC]$, we can define the \emph{division} of monomials with
  $m_2|m_1$ if there exists $m_3\in\cT[\N^n]\times\cT[\cC]$ such that
  $m_1=m_2m_3$. Then, we can make a new S-polynomial of two
  polynomials with supports in  $\cT[\N^n]\times\cT[\cC]$ by considering the $\LCM$
  in $\cT[\N^n]\times\cT[\cC]$ of their leading monomials.
\end{proof}
This lemma shows that the definition of \sgbs and the algorithmic
techniques to compute them in~\cite{FaugereSS2014} can be extended to
skew-polynomial rings $\K[\bt]\ideal{\bx}$ with commutation rules
$t_p x_p=x_p (t_p+1)$.

Using the definitions and notation of Section~\ref{sss:gln}, we have the
following lemma.
\begin{lemma}\label{lem:group_skewpol}
  Let $G$ be a finite group of diagonal matrices acting on
  $t_1,\ldots,t_n,\allowbreak x_1,\ldots,\allowbreak x_n$,
  then $G$ leaves $t_1,\ldots,t_n$, each, invariant.

  Assume that $f_1,f_2\in\K[\bt]\ideal{\bx}$ are both
  $G$-homogeneous polynomials,
  then their S-polynomial
  is also $G$-homogeneous. Thus, so are all the elements of a reduced \gb of
  $\ideal{f_1,f_2}$.
\end{lemma}
\begin{proof}
  There exists a root of
  unity $\zeta$ such that for each matrix in $G$, there exist integers
  $\tau_1,\ldots,\tau_n,\allowbreak \veps_1,\ldots,\veps_n$ such that for all
  $1\leq p\leq n$, $x_p$ is sent onto $\zeta^{\veps_p}x_p$ and $t_p$
  onto $\zeta^{\tau_p}t_p$.

  Therefore, $t_p x_p-x_p t_p=x_p$ is sent on both
  $\zeta^{\tau_p}t_p\zeta^{\veps_p}x_p-\zeta^{\veps_p}x_p\zeta^{\tau_p}t_p
  =\zeta^{\tau_p+\veps_p}\pare{t_p x_p-x_p t_p}=\zeta^{\tau_p+\veps_p}x_p$
  and $\zeta^{\veps_p}x_p$. Thus, $\zeta^{\tau_p}=1$ and $G$ lets
  $t_p$ invariant. By Definition~\ref{def:G-deg}, this means that the
  $G$-degree of 
  $t_p$ is $0$ so that the $\bt^{\bk}\bx^{\bi}$ and $\bx^{\bi}$ have
  same $G$-degree.

  The S-polynomial of $f_1$ and $f_2$ is
  $f_1\bt^{\bk}\bx^{\bi}-f_2\frac{\LC(f_1)}{\LC(f_2)}\bt^{\bl}\bx^{\bj}$
  with
  $\bt^{\bk}\LM(f_1)\bx^{\bi}
  =\bt^{\bl}\LM(f_2)\bx^{\bj}=\GCD\pare{\LM(f_1),\LM(f_2)}$,
  where $\LC(f)$ stands for leading coefficient of $f$, \ie the
  coefficient of $\LM(f)$. Since
  both terms of the sum have the same leading monomial, it remains to
  show that multiplying a polynomial by a monomial preserves the
  $G$-homogeneity. Since
  $\bt^{\bl}\bx^{\bj}\bt^{\bk}\bx^{\bi}
  =\bt^{\bl}\pare{\bt-\bj}^{\bk}\bx^{\bj+\bi}$,
  it is a $G$-homogeneous polynomial of same $G$-degree as
  $\bx^{\bj+\bi}$. Now, the $G$-degree of $\bx^{\bj+\bi}$ is the sum of the
  $G$-degrees of $\bx^{\bj}$ and $\bx^{\bi}$ and thus of
  $\bt^{\bl}\bx^{\bj}$ and $\bt^{\bk}\bx^{\bi}$.
\end{proof}
From Lemmas~\ref{lem:sgb_skewpol} and~\ref{lem:group_skewpol}, we can
compute a \gb or a \sgb of the ideal spanned by skew-polynomials
associated to \Prels in $\K[\bt]\ideal{\bx}$,  with commutation rules
$t_p x_p=x_p (t_p+1)$, to guess new \Prels in
the cone and lattice settings.

\begin{corollary}
  Let $G$ be a finite diagonal matrix group acting on variables
  $\bt$ and $\bx$. Let
  $I=\ideal{f_1,\ldots,f_s}\subset\K[\bt]\ideal{\bx}$ be an ideal
  spanned by $G$-homogeneous polynomials. Then, one can compute a \gb
  of $I$ by using a quasi-commutative variant of the \Fquatre
  algorithm~\cite{Faugere1999}
  building $|G|$ Macaulay matrices for each $G$-degree.
\end{corollary}


\section{Experiments}\label{s:bench}
In this section, we report on our implementations of the
different algorithms of this paper.

We start with the \sFGLM algorithm on a cone, as in
  Subsection~\ref{ss:guessing_cone}, in particular for guessing \Prels of
tables in \textsc{Maple}~2019. This is an extension of
Algorithm~\ref{algo:sFGLM}, see~\cite{BerthomieuF2016}.
We investigate Gessel planar walk $\bgessel$ in the nonnegative
quadrant $\N^2$
with steps in $\{(1,0),(1,1),(-1,0),\allowbreak (-1,-1)\}$ and the
\textsc{3D}-space Walk-$43$ $\bseqq$ of~\cite{BostanBMKM2016} in the
nonnegative octant $\N^3$ with steps in
$\{(-1,-1,-1),\allowbreak (-1,-1,1),\allowbreak (-1,1,0),\allowbreak
(1,0,0)\}$. In particular, we restrict ourselves to a
subsequence of each where one index is $0$. These walks come
naturally with a cone 
structure: for instance whenever $n\neq 2n'+2j$, then
$\gessel_{n,0,j}=0$.
Likewise, whenever $n\neq 8n'+2j+4k$, then
$\seqq_{n,0,j,k} = 0$. Thus, it makes sense to look for the relations
given by the table terms $\gessel_{2 n'+2 j,0,j}$ and $\seqq_{8 n'+2 j+4 k,0,j,k}$.

In Table~\ref{tab:Prel}, we report on the number of computed relations
and the number of relations that do not fail after further testing.
\begin{enumerate}
\item The column Full Orthant means that we consider all the
  table terms $\gessel_{n,0,j}$ and $\seqq_{n,0,j,k}$.
\item The column Half Orthant means that we consider all the
  table terms $\gessel_{2 n',0,j}$ and $\seqq_{2 n',0,j,k}$.
\item The column Cone means that we consider 
  all the table terms $\gessel_{2 n'+2 j,0,j}$ and terms
  $\seqq_{8 n'+2 j+4 k}$, corresponding to the potential nonzero terms.
\end{enumerate}

We
tested two kinds of matrices: matrices that are almost square, with just a
little bit more rows than columns, and matrices that have many more rows
than columns.

We
can notice, as expected in both cases, that by considering only terms on the
nonzero cone we guess many fewer false positive \Prels.
This happens despite our
matrices having fewer rows in the cone setting than in the full
orthant setting, \ie a priori the relations have fewer constraints.
This means that amongst these constraints more are linearly
independent and that in general the number of linearly
dependent rows
is 
responsible for the matrix rank decrease. As a byproduct, this reduces
the number of operations.
\begin{table}[htbp!]
  \centering
  \hspace*{-2cm}
  \setlength\tabcolsep{3pt}
    \begin{tabular}{|l|c|r|r|r|c|r|r|r|c|r|r|r|
      }
      \hline
      Type
      &\multicolumn{4}{c|}{Cone} 
      &\multicolumn{4}{c|}{Half Orthant}
      &\multicolumn{4}{c|}{Full Orthant}\\
      &\multirow{2}{*}{Matrix size}	&\multirow{2}{*}{Queries}
      &\multicolumn{2}{c|}{Relations}
      &\multirow{2}{*}{Matrix size}	&\multirow{2}{*}{Queries}
      &\multicolumn{2}{c|}{Relations}
      &\multirow{2}{*}{Matrix size}	&\multirow{2}{*}{Queries}
      &\multicolumn{2}{c|}{Relations}
      \\
      &&&Fake	&Correct &&&Fake	&Correct	&&&Fake	&Correct 
      \\
      \hline
      $g_{n,0,j}$
      &$\phantom{1\,}444\times \phantom{1\,}441$
      &$\phantom{1\,}866$	&$11$	&$0$
      &$\phantom{1\,}444\times \phantom{1\,}443$
      &$\phantom{1\,}857$	&$68$	&$0$
      &$\phantom{1\,}496\times \phantom{1\,}495$
      &$\phantom{1\,}946$	&$48$	&$0$
      \\
      $g_{n,0,j}$
      &$\phantom{1\,}631\times \phantom{1\,}564$
      &$1\,174$	&$0$	&$0$
      &$\phantom{1\,}961\times \phantom{1\,}581$
      &$1\,506$	&$115$	&$0$
      &$1\,326\times \phantom{1\,}661$
      &$1\,942$	&$84$	&$0$
      \\
      $g_{n,0,j}$
      &$\phantom{1\,}721\times \phantom{1\,}711$
      &$1\,408$	&$15$	&$8$
      &$\phantom{1\,}724\times \phantom{1\,}713$
      &$1\,401$	&$87$	&$0$
      &$\phantom{1\,}726\times \phantom{1\,}715$
      &$1\,386$	&$67$	&$0$
      \\
      $g_{n,0,j}$
      &$1\,951\times 1\,089$
      &$3\,010$	&$0$	&$21$
      &$2\,209\times 1\,036$
      &$3\,196$	&$154$	&$0$
      &$2\,556\times 1\,001$
      &$3\,491$	&$136$	&$6$
      \\
      \hline
      $\seqq_{n,0,i,j}$
      &$\phantom{1\,}223\times \phantom{1\,}211$
      &$\phantom{1\,}430$	&$7$	&$1$
      &$\phantom{1\,}222\times \phantom{1\,}211$
      &$\phantom{1\,}411$	&$25$	&$0$
      &$\phantom{1\,}220\times \phantom{1\,}210$
      &$\phantom{1\,}395$	&$24$	&$0$
      \\
      $\seqq_{n,0,i,j}$
      &$\phantom{1\,}444\times \phantom{1\,}253$
      &$\phantom{1\,}552$	&$2$	&$1$
      &$\phantom{1\,}520\times \phantom{1\,}260$
      &$\phantom{1\,}758$	&$40$	&$0$
      &$\phantom{1\,}680\times \phantom{1\,}267$
      &$\phantom{1\,}912$	&$37$	&$0$
      \\
      $\seqq_{n,0,i,j}$
      &$\phantom{1\,}406\times \phantom{1\,}400$
      &$\phantom{1\,}799$	&$11$	&$6$
      &$\phantom{1\,}406\times \phantom{1\,}400$
      &$\phantom{1\,}772$	&$40$	&$0$
      &$\phantom{1\,}406\times \phantom{1\,}400$
      &$\phantom{1\,}771$	&$27$	&$0$
      \\
      $\seqq_{n,0,i,j}$
      &$\phantom{1\,}806\times \phantom{1\,}522$
      &$1\,320$				&$2$	&$6$
      &$1\,200\times \phantom{1\,}550$
      &$1\,716$	&$78$	&$0$
      &$1\,540\times \phantom{1\,}589$
      &$2\,073$	&$68$	&$0$
      \\
      \hline
    \end{tabular}
  \caption{\label{tab:Prel}\mbox{Guessing fake and correct \Prels with
      Algorithm~\ref{algo:sFGLM} for \Prels on a cone.}}
\end{table}

In Table~\ref{tab:FGLM}, we consider the \FGLM
application, presented in Subsection~\ref{ss:app_fglm}, running on an
\textsc{Intel Xeon} E-2286M with $32$ \textsc{GB}
of \textsc{RAM}. We compute first
a $\DRL$ \gb of an ideal invariant by the action of a finite
diagonal group $\Z/n\Z$
and then the eliminating polynomial of the last
variable. The number $n$ in the names of the systems denotes the number
of variables and the computations were done modulo $2^{30}<p<2^{31}$ such
that a primitive $n$th root of unity exists in $\Z/p\Z$.
The \spFGLM
algorithm~\cite{FaugereM2011,FaugereM2017} has been implemented in
\texttt{C}, as part of the \msolve library~\cite{msolve,msolveweb}, it
generates a scalar table first and then guesses its \Crel with
the \bmy algorithm. Notice that the table generation
is the bottleneck of the method, but it is also the part that
  benefit the most from the occurred speedup.
In the column \spFGLM, we use the whole
multiplication matrix,
while in the column \textsc{lattice} \spFGLM, we use the
$n$ nonzero blocks of the multiplication matrix to perform the
computations and taking advantage of the action of $\Z/n\Z$. We also
compare with \textsc{Maple}~2019 where we use
\texttt{Groebner:-FGLM} to compute a \gb for an ordering
eliminating all the variables but the last one. As expected by
Proposition~\ref{prop:FGLM}, using the
splitting of the multiplication matrix allows us to divide the
computation time by $n$.
\begin{table}[htbp!]
  \centering
  \setlength\tabcolsep{2pt}
    \begin{tabular}{|l|r|r|r|r|r|r|r|}
      \hline
      Type\hfill Degree
      &\multicolumn{2}{c|}{\spFGLM} 
      &\multicolumn{2}{c|}{\textsc{lattice} \spFGLM} 
      &\multicolumn{2}{c|}{\textsc{lattice} speedup} 
      &\textsc{Maple}\\
      &Seq.\ gen.	&Guess.
      &Seq.\ gen.	&Guess.
      &Seq.\ gen.	&Guess.
      &\\
      \hline
      Cyclic-$6$\ \hfill $156$
      &$1\,470$	&$10$
      &$200$	&$2.3$
      &$7.35$	&$4.35$
      &$120\,000$\\
      \hline
      Cyclic-$7$\ \hfill $924$	
      &$64\,000$	&$56$
      &$5\,200$		&$8.3$
      &$12.3$		&$6.75$
      &$13\ s$\\
      \hline
      Random-$3$\ \hfill $294$
      &$3\,100$		&$18$
      &$1\,100$		&$6.8$
      &$2.82$		&$2.65$
      &$510\,000$\\
      \hline
      Random-$3$ bis\ \hfill $3090$
      &$470\,000$	&$170$
      &$83\,000$	&$63$
      &$5.66$		&$2.70$
      &--\\
      \hline
      Random-$4$\ \hfill $896$
      &$69\,000$	&$53$
      &$8\,600$		&$14$
      &$8.02$		&$3.79$
      &$2\,000\ s$\\
      \hline
      Random-$5$\ \hfill $2000$
      &$386\,000$	&$110$
      &$35\,000$	&$24$
      &$11.0$		&$4.58$
      &$49\ s$\\
      \hline
      Random-$6$\ \hfill $1656$
      &$330\,000$	&$91$
      &$26\,000$	&$17$
      &$12.7$		&$5.35$
      &$1\,200\ s$\\
      \hline
      Random-$10$\ \hfill $4160$
      &$13\ s$		&$250$
      &$37\,000$	&$26$
      &$351$		&$9.62$
      &--\\
      \hline
    \end{tabular}
    \caption{\label{tab:FGLM} \FGLM application
      with the action of $\Z/n\Z$ (in
      $\mu s$).}
\end{table}


\section*{Acknowledgments}
We thank the anonymous referees for their careful reading and their helpful
comments to improve this paper.
The authors are supported by the joint \textsc{ANR-FWF}
\textsc{ANR-19-CE48-0015} \textsc{ECARP} project, the \textsc{ANR}
grants \textsc{ANR-18-CE33-0011}
\textsc{Sesame} and \textsc{ANR-19-CE40-0018} \textsc{De Rerum Natura}
projects, the \textsc{PGMO} grant \textsc{CAMiSAdo}, grant
\textsc{FA8665-20-1-7029} of 
the \textsc{EOARD-AFOSR} and and the European
Union's Horizon 2020 research and innovation programme under the Marie
Sk\l{}odowska-Curie grant agreement N.~813211 (\textsc{POEMA}).




\bibliographystyle{elsarticle-harv}
\bibliography{biblio}

\end{document}